\documentclass[11pt,a4paper]{article}

\usepackage[margin=1in]{geometry}
\usepackage{amsmath}
\usepackage{amssymb}
\usepackage{amsthm}
\usepackage{array}
\usepackage{booktabs}
\newcolumntype{L}[1]{>{\raggedright\arraybackslash}p{#1}}
\usepackage{flafter}
\usepackage{placeins}
\usepackage[numbers]{natbib}
\usepackage[hidelinks,unicode]{hyperref}
\newcommand{\SourceDigest}{6696cc2cfe4f8505a61fdcb43bfe035338022a0c2b37842c2370a39c8df310f9}
{unrecorded}}
\hypersetup{
  pdflang={en-US},
  pdftitle={Certified Splitting: Decidability and Anchor Supply for Maximal-Munch Tokenization},
  pdfauthor={Nicklas Nidhögg},
  pdfsubject={Maximal munch tokenization; parallel scanning; certified anchors; decidability},
  pdfkeywords={tokenization, maximal munch, parallel scanning, certification, segmentation, source-sha256:\SourceDigest}
}
\usepackage{etoolbox}
\apptocmd{\thebibliography}{\widowpenalty=10000 \clubpenalty=10000}{}{}

\theoremstyle{definition}
\newtheorem{definition}{Definition}

\theoremstyle{plain}
\newtheorem{theorem}{Theorem}
\newtheorem{lemma}{Lemma}
\newtheorem{corollary}{Corollary}
\newtheorem{proposition}{Proposition}

\newcommand{\code}[1]{\texttt{#1}}
\newcommand{\zenodo}[1]{{\urlstyle{rm}\href{https://doi.org/#1}{\nolinkurl{#1}}}}

\title{Certified Splitting:\\Decidability and Anchor Supply for\\Maximal-Munch Tokenization}
\author{Nicklas Nidh\"ogg\\[0.4em] \normalsize Independent Researcher\\ \normalsize
\texttt{nicklas.nidhogg@gmail.com}\\ \normalsize ORCID: 0009-0006-6161-6150}
\date{October 4, 2026}

\begin{document}
\maketitle

\begin{abstract}
A certificate is a window of bytes with an origin inside it such that, wherever the window occurs in a completely
tokenizable input, the token covering the window's final byte begins at that origin. Certification is decided from the
token set alone, and a certificate's occurrences mark where work may soundly begin. This paper decides certification and
applies it to cuts, edits and audits. It is decided two ways, online with a completeness cutoff for literal vocabularies
and offline through an armed-run verifier for arbitrary regular token sets. Either route, under a declared window
budget, yields the anchor inventory, with a counterexample on each refutation. Cutting a completely tokenizable input at
certified anchors and scanning the chunks independently reproduces the sequential segmentation, with no speculation or
fixup pass. A tokenizability-preserving byte substitution moves boundaries only strictly between the certified anchors
witnessed in the unchanged bytes on either side of it. For tokens of length at most $L$, an edit at $p$ moves no
boundary at or below $p - L + 1$. Cutting at every occurrence of a delimiter is sound exactly when it occurs only at
token starts, and cutting after every occurrence exactly when it occurs only at token ends. Two instruments follow. The
anchor-free span of a finite certified inventory is decided bounded or not with its exact supremum, three for the UTF-8
shape; and a differential auditor decides whether two token sets place different boundaries on an input both tokenize,
with a differing input. Every theorem carries an executable check. Every measurement is a named program's output or
arithmetic on such outputs, and each program ships with the paper or is pinned by commit in the munch library. That
library's threaded scanner splits the archived campaign corpus at certified anchors and reproduces the sequential
boundary stream byte for byte.
\end{abstract}

\section{Introduction}\label{sec:intro}

A certificate is a window of bytes with an origin inside it such that, wherever the window occurs in a completely
tokenizable input, the token covering the window's final byte begins at that origin. A companion paper on
recovery~\citep{nidhogg2026panicmode}, called the recovery companion below, proved such certificates as instruments for
resuming after damage. There, every prefix repair whose scan commits through the certificate's returned evidence places
a boundary at the certified position, so a scanner that has lost its place may trust that position. The certificate
beneath that result quantifies over every completely tokenizable context containing its evidence, and that quantifier
does more work than recovery asks of it. A position that every context agrees is a boundary is not merely a safe place
to resume after something went wrong. It is a safe place to begin, before anything has gone wrong at all.

Two tokens show why recognizing such a position needs a theorem at all. Over $T = \{\code{ab}, \code{b}\}$, a parallel
splitter that cuts at every occurrence of \code{b}, on the ground that \code{b} is a token and so surely starts one,
cuts the input \code{ab} at position 1. The sequential scan commits the single token \code{ab} and has no boundary
there, and the chunk before the cut, the lone byte \code{a}, matches no token. The split scan fails on an input the
sequential scan consumes. The cut is not suboptimal; it is unsound, and no amount of testing on inputs where \code{b}
happens to stand alone will say so. The certificate replaces that intuition with a condition quantified over every
completely tokenizable input. \code{b} occurs at offset 1 of \code{ab}, so some completely tokenizable input places it
past the start of a token, while \code{a} occurs in the token list only at offset zero and so begins a token in every
such input. Cutting at occurrences of \code{a} is sound, and Lemma~\ref{lem:interior} turns that observation into a
decision over the token list alone.

The paper studies one property of a token set through four questions. Does a proposed window force the token covering
its final byte to begin at its named origin wherever the window occurs, whatever the surrounding input, and what is the
exact inventory of such windows through a declared budget $H$? That property is the recovery companion's certificate.
How is such a position recognized, when its definition quantifies over inputs no finite enumeration can exhaust? What
does recognition cost, once that quantifier is discharged? And how often do such positions occur in the material a
scanner will actually meet? The four are not independent, since a useful certificate needs a decision procedure and a
cost bound, and it must also occur in the inputs being scanned. Wherever a scanner starts mid-input, splits work, joins
a stream late, or updates its tokenization after an edit, the question underneath is one of these four. This paper
answers all four for maximal-munch scanning, the first by deciding each proposed window and each exact inventory through
a declared budget, the second by theorem, the third by search-space bounds read off the constructions, and the fourth by
measurement on named vocabularies and grammars.

The two predecessors, the earlier byte- and window-certificate papers, introduce the first question and open the second.
The predecessors, the recovery companion and this paper share one proof pattern, local evidence that the surrounding
context cannot invalidate. The evidence secures a boundary in the predecessors, a resumption in the recovery companion
and a partition here. Certified positions were introduced as single bytes, with a controlled study in which eight of
fifteen token sets, the conventional C-like and JSON ones among them, certify no useful byte, that is, no byte occurring
in a completely tokenizable input~\citep{nidhogg2026splitpoints}. They were generalized to bounded windows that recover
a token origin where no single byte certifies~\citep{nidhogg2026splitwindows}. The first gave an exact criterion for
single-byte windows and the second, in the version cited, a sufficient condition for wider ones. Neither gave a
procedure deciding an arbitrary proposed window, and the window paper bounded its conservative model's search for a
certified window but did not compute the exact inventories studied here. This paper works with the same certificate,
reached through the recovery companion's repair-invariance argument, and adds decision procedures, search-space bounds
and measurements of anchor supply.

The contributions are:

\begin{itemize}
\item An explicit decision procedure for certification over arbitrary regular token sets, the armed-run verifier
(Section~\ref{sec:deciding}), with the underlying decidability credited to Karttunen's construction in
Section~\ref{sec:related}.
\item Exact anchor inventories through a declared window budget, with a synthesized counterexample on each refutation
(Section~\ref{sec:deciding}).
\item A decision whether the anchor-free span is bounded, with its exact value when it is (Section~\ref{sec:deciding}).
\item The differential auditor over pairs of token sets (Section~\ref{sec:diff}).
\item The measured anchor supply of maximal-munch vocabularies derived from byte-pair encoding (BPE,
Section~\ref{sec:measurements}).
\end{itemize}
None of these appear in the predecessors' cited versions, except that the byte paper's cited version builds on the
armed-run automaton, and the window paper's later version takes up the verifier and the supply study's vocabularies from
this paper.

Table~\ref{tab:results} lists what the paper proves and measures, and where. The main results are the decisions, and the
applications are their consequences. None of them requires a bounded anchor-free span. The span decision tells whether
the grammar and chosen inventory guarantee bounded gaps, and the measured anchor supply on a corpus is a fact about that
input.

\begin{table}[!htbp]
\centering
\small
\begin{tabular}{L{2.9cm}L{8.4cm}L{3.3cm}}
\toprule
Result & Statement & Where \\
\midrule
Certification decided & Certification is decided online with a completeness cutoff for literal token sets and offline
through an armed-run verifier for arbitrary regular token sets. Each window and origin reduces to one union-emptiness
problem with a counterexample on refutation, so inventories through a declared budget are computable &
Theorems~\ref{thm:cutoff} and~\ref{thm:offline}, Corollary~\ref{cor:inventory} \\
\addlinespace
Splitting & Cutting a completely tokenizable input at certified anchors and scanning the chunks independently reproduces
the sequential segmentation, for literal and for regular token sets. The split scan was also run on the complete
archived corpora of the recovery campaign's four rows & Lemma~\ref{lem:cut}, Theorem~\ref{thm:split},
Corollary~\ref{cor:splitregular},
Section~\ref{sec:measurements} \\
\addlinespace
Edit damage & A tokenizability-preserving byte substitution moves no boundary outside the open interval between the
certified anchors witnessed in the unchanged bytes on either side of it, nor at or below $p - L + 1$ when the edit is at
$p$ and every token has length at most $L$ & Theorem~\ref{thm:edit} \\
\addlinespace
Delimiters & Splitting at every occurrence of a delimiter is sound exactly when every occurrence of it inside every
token is token-initial, and splitting after every occurrence exactly when every such occurrence is token-final, so a
delimiter whose cuts would split a token is refused before any parallel splitter is built &
Theorem~\ref{thm:delimiter} \\
\addlinespace
Anchor-free span & The span decision tells whether a grammar's anchor-free span under a given inventory is bounded and
computes its exact value when it is, giving three for the UTF-8 shape & Theorem~\ref{thm:density} \\
\addlinespace
Differential auditor & The auditor decides whether two regular token sets differ in boundaries or domain and, when
they do, returns either an input both tokenize on which they place a boundary differently or an input exactly one of
them tokenizes &
Theorem~\ref{thm:diff} \\
\addlinespace
Supply & Anchors per KiB, gap tails and edit hulls for two BPE-derived maximal-munch vocabularies on held-out
slices of the simdjson fixture, and for the campaign's four C-like grammars on their complete archived corpora. Every
reported quantity is emitted by a shipped program or derived by the release gate from named emissions &
Section~\ref{sec:measurements} \\
\bottomrule
\end{tabular}
\caption{The principal contributions and where each is stated, proved or measured. The proofs establish the universal
claims, while source comments identify the executable checks on finite instances and the release gate checks the
manuscript against named emissions.}
\label{tab:results}
\end{table}

Three points of scope hold throughout. First, every certification decision takes a proposed window as input, and every
inventory computation also takes a declared finite budget $H$. Existence at unrestricted width is a separate question,
and Corollary~\ref{cor:existence} decides it both for windows required to occur in some completely tokenizable input and
without that requirement. Second, every supply and edit measurement is a fact about named material under a named policy
and inventory budget, and the edit measurements are also tied to a named seed. The scanner timings are facts about a
named machine as well, and the boundedness verdicts on the anchor-free span say when no grammar-global gap guarantee
exists. Third, the two subword vocabularies are BPE-derived token sets read under maximal munch, which is the scan the
theorems speak about and not the ranked merge procedure that produced them, and the manuscript says so wherever the
distinction bears on a number.

\FloatBarrier
\section{Setting}\label{sec:setting}

This section fixes the scan the paper reasons about, the boundaries it commits, and the certificate the later sections
decide. A \emph{token set} $T$ is a finite nonempty set of nonempty byte strings, and $L$ denotes the maximum token
length. The \emph{maximal-munch scan} of an input $x$ proceeds from position $0$, and at position $c$ it commits the
longest member of $T$ that is a prefix of $x[c..]$, advances past it, and fails if no member matches. For distinct
literal tokens the longest match is unique, so the scan is deterministic. An input is \emph{completely tokenizable} when
the scan consumes it whole. Its \emph{segmentation} is the resulting sequence of committed tokens, and its
\emph{boundaries} are the committed tokens' start positions.

A \emph{regular token set} is a finite nonempty ordered family of empty-string-free regular languages over bytes,
scanned the same way. At each position the scan commits a longest word of any of the languages that prefixes the
remaining input, breaking equal-length ties by list priority, so the committed token is single-valued. A literal token
set is the case whose languages are singletons. Two facts about the scan are used throughout and are immediate from the
definition.

\begin{lemma}[Memorylessness]\label{lem:memoryless}
If $b$ is a boundary of a completely tokenizable $x$, then the scan of the suffix $x[b..]$ taken as an input of its own
commits exactly the tokens the scan of $x$ commits from $b$ on.
\end{lemma}

\begin{proof}
The scan's choice at any position depends only on the input from that position onward, and the scan of $x$ reaches $b$
with no other state than the position itself.
\end{proof}

\begin{lemma}[Boundaries confine matches]\label{lem:confine}
If $b$ is a boundary of a completely tokenizable $x$ and $c < b$ is a boundary, then the token committed at $c$ ends at
or before $b$.
\end{lemma}

\begin{proof}
Committed tokens partition $x$, so the token at $c$ ends at the next boundary after $c$, which is at most $b$.
\end{proof}

A byte $\beta$ is \emph{certified} for $T$ when in every completely tokenizable input every occurrence of $\beta$ begins
a committed token. A window $(W, o)$, with $0 \le o < |W|$, is a \emph{certified window} when, at every occurrence of
$W$ in every completely tokenizable input, the token covering the occurrence's final byte begins exactly $o$ bytes into
the occurrence. The quantifier is per occurrence, so every occurrence must satisfy the condition. The certificate is
sufficient for a sound cut and stronger than a boundary at the origin. Over $T = \{a, b\}$ every occurrence of $ab$ may
be cut at offset $0$, since every byte is a token, yet $(ab, 0)$ is not certified, because the token covering its final
byte begins one byte in, and the decider of Section~\ref{sec:deciding} says so with $ab$ as its counterexample. Which
occurrence-based cut rules beyond this certificate are always sound is not decided here.

A \emph{certified anchor} of a particular input consists of a boundary position and an occurrence of a certified byte or
window that forces that boundary, at the byte itself or at the window's origin. The witnessing occurrence is part of the
anchor, because the edit theorem below conditions on where it lies. The recovery companion proves these certificates
sound for recovery, reading them after damage, while this paper reads the same definitions forward. Certification is
decidable, which Section~\ref{sec:deciding} establishes next, and Section~\ref{sec:splitting} draws the consequences of
a certificate once decided.

\section{Deciding certification}\label{sec:deciding}

A certificate quantifies over all inputs, and this section gives finite procedures that decide that condition, with
compact proofs. Every construction behind the fixed-pair decisions, the finite inventories, the anchor-free span and the
auditor of Section~\ref{sec:diff} is implemented, except the unrestricted-width corollary, whose factors-and-complement
construction is proved and not built. Each lemma is also checked by an exhaustive small-universe sweep in the
repository's exploration programs. Those sweeps, and every finite exhaustive sweep this paper reports, test the
implementations and the worked instances over the universe each one declares. The written proofs establish the general
theorems, and no count in this paper is offered as evidence for a statement quantifying over all token sets and all
inputs.

\begin{lemma}[Interior-byte characterization]\label{lem:interior}
For a regular token set, a byte is certified exactly when every occurrence of it in a word of any token language is
initial.
\end{lemma}

\begin{proof}
If some token $t$ carries the byte at a noninitial offset, the input $t$ alone is completely tokenizable as the single
token $t$, whatever the priority, since no longer token is a prefix of $t$. That occurrence does not begin a token.
Conversely, every occurrence in a tokenizable input lies inside some committed token at the offset it has in that token,
so if every such offset is zero, every occurrence begins a token.
\end{proof}

\paragraph{Convention: non-nullable token languages.} Every regular token set in this paper is empty-string-free, in
that no token language contains the empty string, so every token language is non-nullable and every match consumes at
least one byte. A nullable token does not break the choice, since finitely many token languages still have a longest
matching prefix at every position and the empty match is one candidate among them, so the greedy rule remains well
defined pointwise. It breaks progress. A committed zero-length match neither advances the scan nor closes a segment, so
the process need not terminate, and the marking representation loses its correspondence with positions because a
boundary can repeat without any byte between. This standing convention is fixed in the Setting and applies to every
regular token set, so individual statements do not repeat it. The literal token sets of Section~\ref{sec:setting} are
sets of nonempty byte strings by definition and satisfy it already. Ties between distinct token languages matching the
same longest prefix resolve by list priority, and the reference scans implement that rule.

Wider windows need an automaton model of the scan, built first for a literal token set. Feed the input byte by byte and
keep the \emph{pending tail}, the suffix not yet consumed by committed tokens. The commit of a token is final exactly
when the tail cannot be extended into a strictly longer token.

\begin{lemma}[Pending-state invariant]\label{lem:pending}
An unresolved pending tail is a proper prefix of some token. Hence the online scan is a deterministic automaton whose
states are the token prefixes, at most the sum of token lengths many, plus a dead state. End of input resolves the tail
by a plain scan of it, so acceptance is the predicate that the tail scans clean.
\end{lemma}

\begin{proof}
If the tail is not a proper prefix of any token, no longer token can overtake the longest match at its front, so that
match is final, and either it commits and shortens the tail, or nothing matches and no extension ever can, the dead
state. The commit loop repeats until the tail is unresolved or empty.
\end{proof}

\begin{theorem}[Window certification decides, with a completeness cutoff]\label{thm:cutoff}
Let $T$ be a token set with maximum token length $L$, and let $(W, o)$ be a window with $|W| = h$. A violating
completely tokenizable input, a counterexample to the certificate, exists if and only if one of length at most
$N = h + 2L - 2$ does. Certification of $(W, o)$ is therefore decidable, and a product of the pending automaton with an
occurrence tracker decides it without enumerating inputs.
\end{theorem}

\begin{proof}
The proof truncates a violating input at both ends and needs no pumping argument.

\emph{The head.} Let $x$ violate, with the occurrence at $p$, and let $b$ be the greatest boundary of $x$ at or below
$p$. By Lemma~\ref{lem:memoryless} the suffix $x[b..]$ is completely tokenizable and its segmentation is the
segmentation of $x$ from $b$ on, so the occurrence sits at $p - b$ with the same covering tokens and the violation
survives. The token committed at $b$ covers position $p$, so its length is at least $p - b + 1$, giving
$p - b \le L - 1$, so at most $L - 1$ bytes precede the occurrence.

\emph{The tail.} Work in that suffix, let $t$ be the token covering the occurrence's final byte, and let $e$ be $t$'s
last position. Truncate after $e$. The origin's own position needs no separate accounting, since a token covering it
either ends before the occurrence's final byte or covers that byte as well, and then it is $t$. Every token of the
segmentation beginning at or before $e$ ends at or before $e$, so the truncated input is a whole prefix of that
segmentation, and maximal munch reproduces it. At each position the token the full input's scan takes lies entirely
within the truncated remainder, and deleting a suffix makes no longer token match. The truncated input is therefore
completely tokenizable, still contains the occurrence, and carries the same boundaries across it, so it still violates.
It ends at a token end, so end of input yields a complete scan, even when that token remains ambiguous until end of
input.

\emph{The length.} The occurrence begins at index at most $L - 1$ and is $h$ bytes long, so its final byte has index at
most $L + h - 2$. The token covering that byte begins at or before it and has length at most $L$, so $e \le 2L + h - 3$
and the truncated input has length at most $N = h + 2L - 2$.

Enumerating all inputs through length $N$ and checking their greedy segmentations therefore decides certification, and
the enumeration helper uses this $N$ as its cutoff. Tokens that resolve only at end of input need no separate case. Over
$T = \{\code{ba}, \code{bac}\}$ with $W = \code{a}$ and $o = 0$, the commit of \code{ba} is final only when the input
ends, since one further byte would let \code{bac} overtake it, and the truncation cuts after that token's last position
all the same.

The product procedure decides certification without enumerating inputs. It searches configurations made of a pending
state, the occurrence match's progress and the violation phase, and accepts a violating configuration exactly when its
pending state can still be completed to a completely tokenizable input. The occurrence tracker it carries records the
designated occurrence's final-byte offset within the pending tail, initially at most $L - 1$, decreasing only as bytes
commit, and resolved at a streaming commit or at end of input, so the product has at most $|P(T)|(h + 1)(L + 2)$
configurations, $O(|P(T)| \cdot h \cdot L)$, where $P(T)$ is the set of proper token prefixes, the empty one included.
\end{proof}

% machine-check: applications.py, sweep_cutoff_phases (238,662 violating occurrences over the small
% universe carry the pending-state invariant the online automaton needs, every pending tail shorter
% than the longest token and the covering commit within L bytes of its start; and over
% 3,684 refuted triples the minimal violation fits the derived $|W| + 2L - 2$, tight on at least one
% of them, that constant being the enumeration cutoff itself. explore_certification.py walks both
% truncations on every violation it finds, verifying the designated occurrence where it sits rather
% than requiring it to be the first in its input, and its --wide emission declares the parameters,
% the count of occurrences walked, and how many of them were oversized before truncation.

The cutoff is attained, so at these parameters it is tight. Over $T = \{\code{aab}\}$ with $W = \code{ba}$ and $o = 0$,
the completely tokenizable inputs are exactly the repetitions of \code{aab}, and the shortest one containing \code{ba}
at all is $x = \code{aabaab}$. The occurrence begins at position 2, while the token covering its final byte begins at
position 3, one byte past the origin, so $x$ violates. Here $L = 3$ and $h = 2$, so $N = h + 2L - 2 = 6$, which is
exactly $|x|$, and no shorter tokenizable input contains the window at all.

The online route has a limit of its own. Over $T = \{\code{a}^{+}\code{b}, \code{a}\}$, after reading $\code{a}^{n}$ the
end of the first token is still undecided, falling after the first \code{a} if the input ends there and after the
\code{b} if one follows, so no fixed lookahead settles it, and a subsequential finite-state transducer, which emits its
output irrevocably as it reads, cannot produce the boundary stream at all. Certification remains decidable for this set,
because a verifier can check completed segmentations rather than compute them.

\begin{definition}[Armed-run automaton]\label{def:armed}
Fix a regular token set. The armed-run automaton, the armed-run verifier when it checks a proposed marking, reads an
input together with a boundary marking, one bit per position, from a distinct unconsumed initial state that no consumed
byte revisits. At each mark it starts a fresh run of every token automaton. A run is unarmed until the next mark, where
the segment it has read must be accepted by some token, and armed from then on. An armed accept at any later position
rejects the marking, since it exhibits a longer match the greedy scan would have taken. The automaton accepts when the
final segment closes at end of input and no run has rejected, and it accepts the empty marking in its initial state.
Every nonempty marking it accepts marks position zero, a mark is read before the byte at its position, and end of input
closes the final segment.
\end{definition}

\begin{theorem}[Offline verifiability]\label{thm:offline}
A boundary marking of an input, carrying boundary positions and no token identities, is the boundary projection of the
greedy segmentation of a completely tokenizable input exactly when the armed-run automaton accepts it. The automaton's
state space is finite. Consequently window certification is decidable for arbitrary regular token sets.
\end{theorem}

\begin{proof}
Let the marking cut $x$ into segments $s_1, \dots, s_k$ at boundaries $b_0 = 0 < b_1 < \dots < b_k = |x|$. The empty
input is the case $k = 0$, where no token matches it, the languages being empty-string-free, so the greedy scan commits
nothing and ends complete, and the automaton accepts the empty marking in its initial state. That state is distinct
because the ordinary start configuration can recur. Under the token language $(\code{ab})^{*}\code{a}$, reading
$\code{ab}$ after the initial mark with no further mark restores it, and accepting there would confuse the empty input
with a nonempty one. The finiteness count below includes that state.

\emph{Soundness.} Suppose the automaton accepts. Each closure at $b_i$ required some unarmed run started at $b_{i-1}$ to
accept, so each $s_i$ is a token. We show by induction that the greedy scan of $x$ commits exactly $s_i$ at $b_{i-1}$.
It suffices that no token strictly longer than $s_i$ is a prefix of $x[b_{i-1}..]$. Were $t$ such a token, the run of
$t$'s automaton started at the mark $b_{i-1}$ would still be alive at $b_i$, be armed from that mark on, and accept at
$b_{i-1} + |t| > b_i$, an armed accept, contradicting acceptance. Hence $s_i$ is the longest token prefix, the greedy
scan commits it, and since the final segment closes at $|x|$, the input is completely tokenizable with the accepted
marking as its greedy segmentation.

\emph{Completeness.} The greedy segmentation of a completely tokenizable input passes both tests. Every segment is the
committed token, so every closure succeeds. An armed accept at position $j > b_i$ by a run started at $b_{i-1}$ would
exhibit a token prefix of $x[b_{i-1}..]$ of length $j - b_{i-1} > |s_i|$, contradicting the longest-match choice at
$b_{i-1}$.

\emph{Finiteness.} An ordinary state is a set of configurations, one configuration being a token automaton's derivative
paired with an armed flag, and beside the ordinary states there is the single unconsumed initial state, which no
ordinary state equals. The derivatives are finitely many up to the standard associativity, commutativity, and
idempotence quotient of alternation, so the configurations and the ordinary states are finitely many, and the accepted
marked-string language is regular. Certification of a pair $(W, o)$ is then the emptiness of the intersection of this
regular language with the regular language of markings carrying a badly covered occurrence of $W$, and a finite
reachability search decides that emptiness and returns a violating input when one exists.
\end{proof}

\begin{corollary}[Computable anchor inventories]\label{cor:inventory}
For a regular token set and a declared window budget $H$, the certified pairs $(W, o)$ with $|W| \le H$ are computable
by one union-emptiness problem per pair over the armed-run product, with no string enumeration and no cutoff.
\end{corollary}

\begin{proof}
By Theorem~\ref{thm:offline}, a search over the product walks exactly the greedy segmentations of completely tokenizable
inputs containing the occurrence, and one bit records whether the covering token starts at the origin when the window's
final byte is read. The product state space is finite, so breadth-first search terminates, and the pair is certified
exactly when the search finds no violating accepting run.
\end{proof}

The implementation searches two branches, one for an occurrence at the start of the input and one for a later
occurrence. Each follows valid markings byte by byte, using the armed simulation to reject nongreedy markings. It
chooses an occurrence nondeterministically, tracks its progress through the window, and records in one bit whether the
latest mark is at the origin. A violation exists exactly when a state reached after the window has this bit clear and
can still reach a valid end of input.

% machine-check: the construction itself is not built; it is Theorem~\ref{thm:offline}'s language put
% through the standard regular closure operations, and the left-extension half needs no computation.
% The occurring reading of the campaign certificates is decided by occurring_windows.py and
% reported in the measurement section, whose passage comment names the check.
\begin{corollary}[Existence at unrestricted width is decidable]\label{cor:existence}
For a fixed regular token set the certified pairs form an effectively regular language over the byte alphabet with one
added origin marker, and so do the certified pairs whose window occurs in some completely tokenizable input. Both
languages have decidable emptiness, and a shortest member of either is computable when one exists. Moreover
certification is closed under extension on the left, in that if $(W, o)$ is certified then so is $(UW, |U| + o)$ for
every byte string $U$.
\end{corollary}

\begin{proof}
Encode a candidate as $u \# v$ with $W = uv$, $o = |u|$ and $v$ nonempty, over the byte alphabet extended by one marker
$\#$ that occurs exactly once. Theorem~\ref{thm:offline} gives the regular language $G$ of valid boundary-marked inputs.
The factors of $G$ are regular. The same automaton, with every reachable state admitted as a start and every
co-accessible state as an accept, accepts exactly the marked stretches that occur inside some completely tokenizable
input. Place the marker in such a factor and record, with a finite flag, whether the token covering the factor's last
byte begins at the marker, which holds when the mark on the first byte after the marker is set and every later mark in
the factor is clear. Keep the factors where the flag fails, project away the mark track, and the result is the regular
language $B$ of encoded pairs having at least one badly covered occurrence in some completely tokenizable context.
Complement $B$ within the regular set of well formed encodings $\Sigma^{*}\#\Sigma^{+}$, writing $\Sigma$ for the byte
alphabet, and the complement is the language of certified pairs, regular and effectively constructed. Emptiness,
infinitude and a shortest member are then the standard decisions on a finite automaton.

For the left extension, fix an occurrence of $UW$ in a completely tokenizable input. It contains an occurrence of $W$
ending at the same byte, and the token covering that byte is the same token in both readings. Since $(W, o)$ is
certified, that token begins at offset $o$ within the occurrence of $W$, which is offset $|U| + o$ within the occurrence
of $UW$. Every occurrence of $UW$ arises this way, so $(UW, |U| + o)$ is certified. Prefixing arbitrary strings then
produces certified pairs of every greater width, so the unrestricted certified language is either empty or holds
certified windows of every width from its shortest one on. For the occurring language, the restriction is an
intersection with the marked factor language of $G$'s byte projection, which is regular by the same operations. A
prefixed $UW$ need not occur even when $W$ does, so the left-extension closure is a statement about the unrestricted
reading.
\end{proof}

A window that occurs in no completely tokenizable input is certified vacuously. Thus, if a token set leaves a byte
unused, the unrestricted language is already nonempty at width one. The occurring form of the existence decision
excludes such vacuous certificates. The measurements use inventories through a declared budget and not the
unrestricted-width existence decision, and Section~\ref{sec:measurements} states beside the campaign rows which reading
each inventory reports and which campaign certificates occur.

\begin{theorem}[The anchor-free span is decidable]\label{thm:density}
For a regular token set and a finite certified inventory, define the anchor-free span as the supremum, over completely
tokenizable inputs, of the lengths of their anchor-free stretches, the runs of positions carrying no inventory anchor,
so that it is the longest anchor-free stretch over all of them. Whether this span is finite is decidable, and when it is
finite its value is computable.
\end{theorem}

\begin{proof}
First consider the empty inventory, which the statement admits and which has no longest window for the construction
below to buffer. With no certificate every position is anchor-free, so the question is only whether completely
tokenizable inputs are arbitrarily long, and the union of the token languages settles it in three cases. If that union
is infinite it contains words of unbounded length, each a completely tokenizable one-token input, so the stretch is
unbounded. If it is finite and nonempty, take a word $t$ of maximum length. No token is longer, so the greedy scan of
$t^{n}$ commits $t$ each time, and those inputs are completely tokenizable of unbounded length. If it is empty, only the
empty input is completely tokenizable and the supremum is zero. An infinite regular language need not have a longest
word, which is why the first case cannot be folded into the second.

Assume from here that the inventory is nonempty. Write $M$ for the maximum window length in the inventory. When token
lengths have a finite bound $L$, $M$ need not be bounded by it. Build the \emph{credit product} from the armed-run
verifier of Theorem~\ref{thm:offline}, the last $M-1$ bytes read, and one flag per buffered position recording whether
some completed window has already credited it as an anchor. The product is finite, being a verifier state, a bounded
byte buffer, and a bounded flag vector.

\emph{Flags settle before they leave.} A window $W$ with origin $o$ credits the position $o$ bytes into its occurrence,
and it completes when its final byte is read, which is $|W| - 1 - o$ positions later. That age is at most $M - 1$, so a
position is credited, if at all, no later than the step at which it reaches age $M - 1$. Taking that step's credits
before the position leaves the buffer, it leaves with its flag final, and no later byte can change it.

\emph{Live states only.} Restrict the product to states reachable from the start and from which acceptance is still
reachable. Anchor-free stretches are measured on completely tokenizable inputs, and only live states occur in the runs
of such inputs, so nothing is lost and nothing spurious is admitted.

\emph{The unbounded case.} Suppose the live subgraph, restricted to transitions whose exiting position leaves
unanchored, contains a cycle. Every live state lies on some accepting run, and pumping the cycle inside such a run
yields accepted runs with arbitrarily long anchor-free stretches, so the span is unbounded. Conversely, if the span is
unbounded then some accepting run carries an anchor-free stretch of exits longer than the number of live states, the
buffered tail past the last exit contributing at most $M - 1$ more positions, and the stretch's exits repeat a state,
closing such a cycle. Bounded and acyclic therefore coincide, and cycle detection decides which holds.

\emph{The bounded case.} With that subgraph acyclic, let $r(s)$ be the longest anchor-free run of exits ending at the
live state $s$. The recurrence runs over incoming labeled transitions $s' \to s$, each contributing a candidate, namely
$r(s') + 1$ where the transition's exiting position leaves unanchored, $0$ where it leaves anchored, and $r(s')$ during
warm-up, where the transition exits no position at all. The value $r(s)$ is the largest candidate over all incoming
transitions, so a state carrying both anchored and unanchored incoming edges is not forced to either.

Only the unanchored-exit subgraph is required to be acyclic; anchored exits may close cycles in the whole live graph.
The values $r(s)$ still stabilize under repeated application of the recurrence. A warm-up transition strictly increases
the buffer's occupancy and no transition decreases it, so none lies on a cycle. Every cycle therefore consists of exits
and contains an anchored one, which resets the run to zero, since a cycle of unanchored exits would lie in the subgraph
the bounded case has just made acyclic. Run lengths are thus bounded by the number of live states, and the start state's
value is initialized to zero.

The values $r(s)$ account for the stretches of finalized exits. At end of input, the buffered flags account for the
stretches that extend into, or lie wholly within, the remaining buffer. At an accepting state the trailing unanchored
stretch of exits continues through the oldest contiguously unanchored buffered positions, and a stretch lying wholly
inside the buffer is the longest unanchored run among the flags. The supremum is therefore the largest of three maxima,
of $r(s)$ over live states $s$, of $r(s)$ extended through the oldest contiguously unanchored buffered positions over
accepting states $s$, and of the longest unanchored run among the flags of an accepting state. A prefix attaining the
first maximum extends to an accepting run because its state is live, and the other two maxima are attained on accepting
runs already.
\end{proof}

% machine-check: certified_gap_regex.py, sweep_bound_against_enumeration (four bounded verdicts,
% each equal to the longest anchor-free stretch over every tokenizable string to length twelve,
% and two unbounded verdicts whose measured stretch reaches that length); check_buffer_monotone,
% asserted inside every bound the module computes, is the buffer-occupancy invariant the bounded
% case's termination argument uses.
Over the small universes the bound of every bounded verdict equals the longest anchor-free stretch found by enumerating
all tokenizable strings through length twelve, and every unbounded verdict is matched by a measured stretch that reaches
that length. The enumeration checks the implementation within that range, and the proof above establishes the supremum
over all completely tokenizable inputs.

The UTF-8 shape's inventory yields the bound three, computed from the token set, while a miniature C-like grammar of
identifiers, numbers, spaces, and an operator has an unbounded anchor-free span, because its identifier tokens have
arbitrarily long interiors with no boundary. A grammar gets a synchronization guarantee from its design, and this
theorem decides whether the grammar and chosen inventory give a bounded anchor-free span and computes the bound when
they do. Bounded token length is needed for that guarantee but does not supply it, as Proposition~\ref{prop:length}
proves with the token set $\{\code{aa}\}$ worked in Section~\ref{sec:related}.

\begin{proposition}[Token length and the anchor-free span]\label{prop:length}
For a regular token set, bounded token length is necessary but not sufficient for a bounded anchor-free span. When the
token languages hold words of unbounded length, the anchor-free span is unbounded for every certified inventory, and
over $T = \{\code{aa}\}$, whose one token has length two, it is unbounded for every certified inventory as well.
\end{proposition}

\begin{proof}
Every word $t$ of a token language is a completely tokenizable input on its own, committed whole as the single token $t$
whatever the priority, since no longer token is a prefix of $t$. Its $|t| - 1$ interior positions carry no boundary, so
no certificate makes any of them a boundary and none is a certified anchor, whatever the inventory. The anchor-free span
is therefore at least $|t| - 1$, and words of unbounded length make it unbounded. For the second half, the completely
tokenizable inputs over $\{\code{aa}\}$ are the even-length runs of \code{a}. An occurrence of a window of \code{a}s
beginning at position $p$ has its last byte covered by the token beginning at the greatest even position at or below
that byte, so the covering token's offset within the occurrence depends on the parity of $p$, and both parities occur
once the input is long enough. No occurring window is certified, and a window that occurs in no completely tokenizable
input anchors nothing, so the inputs $\code{a}^{2n}$ carry no anchor at all and make the span unbounded under every
certified inventory.
\end{proof}

With the certified bytes as the inventory, the anchor-free span is the synchronization distance, the longest run of a
valid stream that carries no certified byte, and computing it needs no buffer.

\begin{corollary}[Synchronization distance is computable]\label{cor:gap}
For a token set $T$, the supremum over completely tokenizable inputs of the longest run carrying no certified byte is
computable. It is unbounded exactly when the subgraph of the pending automaton's live states using only non-certified
bytes has a cycle, and otherwise it is the longest path through that subgraph.
\end{corollary}

\begin{proof}
This is Theorem~\ref{thm:density} with the certified bytes as the inventory, or its empty-inventory case when no byte is
certified. With that inventory $M = 1$, the buffer is empty and an exit is credited exactly when its byte is certified.
The form of the statement over the pending automaton of Lemma~\ref{lem:pending} is the direct argument. A run avoiding
certified bytes can begin after any prefix, so every completable state starts one. Each such run follows transitions
labeled by non-certified bytes between completable states, and conversely every such path occurs in some completely
tokenizable input, between a prefix reaching its start state and a completion from its end state. Longest path in a
finite graph is computable, and unboundedness is exactly the presence of a reachable cycle of these transitions.
\end{proof}

The worked example is the UTF-8 shape over byte classes, one ASCII class, three lead classes and one continuation class,
and its tokens are the four sequence shapes. Lemma~\ref{lem:interior} derives the code's well-known
self-synchronization, the certified bytes being exactly the ASCII and lead classes, and Corollary~\ref{cor:gap} computes
the synchronization distance three, recovering the bound implied by the UTF-8 byte forms of
RFC~3629~\citep{yergeau2003utf8}, which lists the self-synchronization property and the one-to-four-byte forms but not
the number three. The name is borrowed from code synchronization, where the delay of a code is counted in codewords and
assumes unique decipherability \citep{restivo1975synchronization}. The distance here is in bytes and is computed for a
token set that need not be a code.

This makes the distance a design-time metric. A proposed format's synchronization distance, computed from its
certified-byte inventory, bounds the length of every run of a completely tokenizable input that carries no certified
byte, a final run ending at the end of input included, so a reader joining a valid stream mid-way reads at most that
many bytes before it meets a certified byte or the end of the stream. The same bound holds after corruption whenever the
damaged stream remains completely tokenizable, and a stream damaged past that point is the subject of the recovery
companion, whose repair model rests on the same certified bytes. Inventory and distance are both computed before the
format ships.

% machine-check: applications.py, utf8_shape and certified_gap_bound.

% machine-checks: explore_certification.py (cutoff sweeps and the two truncations walked over
% every violation found, with data/wide-cutoff-sweep.txt the wider emission the release gate
% reproduces byte for byte),
% pending_automaton.py (product procedure against enumeration, 1,434 triples),
% regex_frontier.py (the online boundary, measured), offline_certification.py (verifier exact on
% 1,020 strings with every candidate marking tried; the two truncations over the unbounded-delay set),
% certified_inventory_regex.py (the corollary's search against bounded enumeration, 136 decisions
% equal across four regex sets; a miniature C-like inventory computed and its anchors splitting a
% generated corpus with chunked scans byte-identical to sequential),
% certified_gap_regex.py (the anchor-free span theorem: utf8 shape bound three reproduced, the miniature
% C-like grammar and the unbounded-delay set both decided unbounded, bounds held against direct
% measurement over generated corpora).

\section{Splitting theorems}\label{sec:splitting}

This section draws the consequences of a certificate, and Section~\ref{sec:measurements} measures them. A cut at
boundaries is exact, so a certified position is a sound cut. A tokenizability-preserving byte substitution moves
boundaries only strictly between the certified anchors witnessed in the unchanged bytes on either side of it. Cutting
before every occurrence of a delimiter is sound exactly when it occurs only at token starts, and cutting after every
occurrence exactly when it occurs only at token ends.

% machine-check: explorations/certification/applications.py, sweep_parallel_split
% (1,554 tokenizable strings exhaustively, chunked equal to sequential on all; corpus-scale
% demonstration in parallel_corpus_demo.py: 8,192 bytes, 2,376 anchors, 8-worker snap within 15
% bytes, segmentation identical).
\begin{lemma}[Cuts at boundaries are exact]\label{lem:cut}
Let $T$ be a token set, $x$ completely tokenizable under $T$, and $B$ a set of boundaries of $x$. Scanning each segment
of $x$ between consecutive members of $\{0\} \cup B \cup \{|x|\}$ independently yields segmentations whose concatenation
equals the sequential segmentation of $x$.
\end{lemma}

\begin{proof}
Fix consecutive cut points $u < v$ from $\{0\} \cup B \cup \{|x|\}$ and consider the chunk $y = x[u..v)$. We show by
induction along the chunk that the chunk-local scan of $y$ commits, position for position, the tokens the sequential
scan commits inside $[u, v)$.

In the coordinates of $x$, both scans stand at $u$, a boundary. Suppose both stand at a common position $c$ with
$u \le c < v$, a boundary of the sequential scan. The sequential scan commits the longest member $t$ of $T$ prefixing
$x[c..]$. The token $t$ ends at or before $v$, trivially when $v = |x|$, and by Lemma~\ref{lem:confine} applied to the
boundary $v$ when $v \in B$. Then $t$ is also a prefix of $y[c-u..]$ and is available to the chunk-local scan.
Conversely every candidate available to the chunk-local scan is a prefix of $x[c..]$ and hence available to the
sequential scan, so no chunk-local candidate is longer than $t$. For distinct literal tokens the longest match is
unique, so the two scans commit the same token and advance to the same position. The induction reaches $v$ exactly,
since the sequential boundaries inside $[u, v)$ partition the chunk, so the chunk-local scan consumes $y$ completely,
and adding $u$ to its boundaries gives the sequential boundaries in $[u, v)$. Concatenating over all chunks yields the
claim.
\end{proof}

\begin{theorem}[Certified positions split soundly]\label{thm:split}
Let $T$ be a token set, $x$ completely tokenizable under $T$, and $A$ a set of certified anchor positions of $x$.
Scanning each segment of $x$ between consecutive members of $\{0\} \cup A \cup \{|x|\}$ independently yields
segmentations whose concatenation equals the sequential segmentation of $x$.
\end{theorem}

\begin{proof}
Every anchor $a \in A$ is a boundary of the sequential segmentation, by the defining guarantee of its certificate
applied to the tokenizable input $x$ itself, so Lemma~\ref{lem:cut} applies.
\end{proof}

Certification establishes from the token set alone that every occurrence of a certified window in a completely
tokenizable input places a boundary at its origin, so a cut set drawn from a certified inventory is exact on every
completely tokenizable input. The preceding results are stated for literal tokens, and the next corollary extends them
to the regular token sets defined in Section~\ref{sec:setting}.

% machine-check: applications.py, sweep_parallel_split, and parallel_corpus_demo.py at scale.

\begin{corollary}[Splitting at the regular generality]\label{cor:splitregular}
Lemmas~\ref{lem:memoryless}, \ref{lem:confine} and \ref{lem:cut} and Theorem~\ref{thm:split} hold verbatim for regular
token sets, so a cut at boundaries is exact there too, and a certified position splits soundly.
\end{corollary}

\begin{proof}
The proofs of the three lemmas use only three properties of the scan, that its choice at a position depends only on the
suffix from that position, that its committed tokens partition the input, and that this choice is the best of the token
members prefixing the remaining input under a fixed preference. Shortening that remainder can remove a candidate but
cannot displace a winner that still fits, and the chunk-local step relies on this.

The third property is not implied by the first two. A rule that commits the shortest match when exactly three bytes
remain, and the longest match otherwise, satisfies the first two, yet its chunk-local scan disagrees with its sequential
one at a cut that is a boundary of the input, over $\{\code{a}, \code{aa}\}$ on \code{aaaaa} cut at $3$. All three hold
for a regular token set as the Setting defines it, the priority order making the choice single-valued, so the arguments
transfer unchanged, and Theorem~\ref{thm:split} follows from Lemma~\ref{lem:cut} as before.

Tie-breaking needs explicit attention because several token languages may match the same longest prefix. The two scans
see the identical family of candidates ending at or before the cut, longest ones included, and resolve a tie among
equal-length candidates by the same list-priority rule, so they commit the same token and not merely the same length.
\end{proof}

For a single-byte substitution at position $p$ and a common length bound $L$ on the token languages, the edit's
lookahead shadow is $(p - L,\, p]$, the start positions whose token choice can depend on the edited byte. A choice made
at or below $p - L$ depends only on unchanged bytes, so no anchor is needed to bound the damage on the left. The first
boundary past $p - L$ is $0$ or ends a token committed at or below $p - L$, so both inputs share it and the damage lies
above $p - L + 1$. Regularity plays no part in this bound, since a common length bound on the token languages is all the
argument uses, and for a regular token set whose languages hold words of unbounded length the argument supplies no such
bound. On either side, a certified anchor witnessed in bytes the edit left untouched also bounds the damage, since its
certificate makes it a boundary of both inputs.

% machine-check: applications.py, sweep_incremental_bound (1,346 tokenizable single-byte edits,
% every moved boundary inside the stated window); sweep_edit_between_anchors (11,926 tokenizable
% single-byte substitutions across five regular token sets, inventories from the emptiness
% decision, every moved boundary strictly between the flanking anchors, and on the 8,046 over the
% three sets whose words are bounded none at or below the lookahead shadow's first position).
\begin{theorem}[Edit damage is confined between witnessed anchors]\label{thm:edit}
Let $T$ be a regular token set, $x$ completely tokenizable under $T$, and $x'$ a single-byte substitution at position
$p$ that is completely tokenizable under $T$. Let $a$ be a certified anchor position of $x$ whose witnessing occurrence
lies wholly in $[0, p)$, or $a = 0$, and $q$ a certified anchor position of $x$ whose witnessing occurrence lies wholly
in $[p+1,\, |x|)$, or $q = |x|$. Every boundary present in exactly one of the two segmentations lies strictly between
$a$ and $q$. When an integer $L$ bounds the length of every word of every token language of $T$, and in particular when
$T$ is a literal token set with maximum token length $L$, every such boundary lies in $(\max(a,\, p - L + 1),\, q)$.
\end{theorem}

\begin{proof}
Low side, from the anchor. When $a = 0$ it is a boundary of both segmentations by definition and nothing lies below it.
Otherwise the occurrence behind $a$ lies in bytes the two inputs share, so its certificate makes $a$ a boundary of both
segmentations. By Corollary~\ref{cor:splitregular} applied to each input with the single cut $a$, each segmentation
restricted to $[0, a)$ is the chunk-local scan of its own first $a$ bytes, and those byte ranges are identical, so the
two segmentations agree on every boundary below $a$, and both carry $a$ itself.

Low side, from the shadow. With $L$ the common length bound on the token languages, the two inputs agree on $[0, p)$,
and we show that the two scans commit identical tokens at every common boundary $c \le p - L$. The choice at $c$ is
decided by the words of length at most $L$ that prefix the remaining input, all inside $[c, c + L) \subseteq [0, p)$,
where the inputs agree, so the same words match, the same one wins under the same priority, and the committed token and
the next boundary coincide. By induction from $0$, the boundary sequences of $x$ and $x'$ are identical up to and
including the first boundary $d$ exceeding $p - L$. Both scans therefore carry $d$, so their boundaries in $[0, d]$
coincide, and $d \ge p - L + 1$. Hence no boundary at or below $p - L + 1$ belongs to exactly one segmentation, and with
the previous paragraph none at or below $\max(a,\, p - L + 1)$ does.

High side. If $q = |x|$, no boundary lies at or beyond $q$ and nothing is claimed there. Otherwise the occurrence behind
$q$ lies in the unchanged suffix $[p+1,\, |x|)$, so it is present in both inputs, and both are completely tokenizable.
The certificate therefore makes $q$ a boundary of both segmentations. By Lemma~\ref{lem:memoryless}, which
Corollary~\ref{cor:splitregular} carries to regular sets, each segmentation from $q$ on is the segmentation of the
suffix from $q$. The two suffixes are byte-identical, so the segmentations agree on every boundary at or beyond $q$.
Every boundary present in exactly one of the two segmentations therefore lies in $(a, q)$, and under the length bound in
$(\max(a,\, p - L + 1),\, q)$.
\end{proof}

The greatest such $a$ and the least such $q$ give the tightest interval. Any qualifying pair gives a valid one, and the
budget-limited inventories of Section~\ref{sec:measurements} report such a pair.

% machine-check: applications.py, ingestion_demo; the refuting half is the token-interior window
% lemma of the decidability note.
\begin{theorem}[The delimiter-splitting criteria, a dual pair]\label{thm:delimiter}
Let $T$ be a regular token set and $\beta$ a byte of its alphabet. Splitting every completely tokenizable input
\emph{at} each occurrence of $\beta$ and scanning the pieces independently preserves the segmentation if and only if
$\beta$ occurs in tokens only at position zero. Splitting \emph{after} each occurrence preserves it if and only if
$\beta$ occurs in tokens only at final positions.
\end{theorem}

\begin{proof}
At-splitting. If $\beta$ occurs in tokens only at position zero, then $\beta$ is a certified byte by
Lemma~\ref{lem:interior}, since every occurrence in a tokenizable input lies inside some committed token at the offset
it has in that token, and every such offset is zero. Theorem~\ref{thm:split}, carried to regular sets by
Corollary~\ref{cor:splitregular}, then gives preservation. Conversely, if some token $t$ carries $\beta$ at offset
$j > 0$, the input $t$ alone is completely tokenizable as the single token $t$, whatever the priority, since no longer
token can be a prefix of $t$. Cutting at every occurrence of $\beta$ in $t$ produces chunks whose cut points include
$j$, and there may be more than two of them when $t$ carries $\beta$ more than once. If any chunk fails to scan, the
split scan fails on an input the sequential scan consumes. If every chunk scans, the concatenated chunk segmentation
contains a boundary at $j$, a token start of its own chunk, that the sequential segmentation of the single token $t$
lacks. Thus preservation fails whether or not all the chunks scan.

After-splitting. If $\beta$ occurs in tokens only at final positions, every occurrence in a tokenizable input is the
last byte of its committed token, so the position after it starts the next token and is a boundary, except when the
occurrence is the input's last byte, where that position is $|x|$ itself. That endpoint is not a token start, but the
cut set of Lemma~\ref{lem:cut} already contains it and it opens no chunk. The lemma, at the regular generality of
Corollary~\ref{cor:splitregular}, then gives preservation. Conversely, if some token $t$ carries $\beta$ at a non-final
offset $j$, the single-token input $t$ is cut at $j + 1$, mid-token, and fails preservation as before.
\end{proof}

% machine-check: applications.py, ingestion_demo and sweep_after_split_criterion (198 pairs, both
% directions exhaustive; the first conjectured criterion was refuted by the sweep and replaced by
% this dual form).

\section{The differential auditor}\label{sec:diff}

This section compares two token sets with the verifier of Theorem~\ref{thm:offline}. Two such verifiers over a common
input decide whether the sets ever place a boundary differently, and a determinized projection of each decides whether
they ever disagree on whether an input tokenizes at all.
\begin{theorem}[Segmentation differentials are decidable for regular token sets]\label{thm:diff}
Given regular token sets $T_1$ and $T_2$, it is decidable whether some input completely tokenizable under both receives
different boundary sets, and likewise whether some input is completely tokenizable under exactly one of the two. In each
case a differing input is computable when one exists.
\end{theorem}

\begin{proof}
Run both of Theorem~\ref{thm:offline}'s armed-run verifiers in product over a common input, each side
nondeterministically guessing its own marking, with a flag set when a position is marked by exactly one side. The armed
pruning confines each side's guesses to its greedy segmentation. Any other guess is rejected by an armed accept or an
unclosable segment, so a run that survives feeds the two true boundary sets position by position. Accept when the input
can end, both sides closing a final segment, with the flag set. Such a run exhibits an input both sets tokenize whose
boundary sets differ, and conversely any such input yields an accepting run by guessing both true markings, the flag set
at the first boundary they dispute. The product state, two derivative-set pairs and one flag, is finite, so the search
terminates, and reconstructing the fed bytes from an accepting run yields the differing input.

The second predicate needs a different construction, because a nondeterministic pair product cannot establish
non-membership on one side. A run failing to accept is not a proof that no run accepts. Project each side's verifier
onto bytes by erasing the mark bits, so that it accepts a byte string whenever some marking of it is accepted, which is
exactly when that side tokenizes the string completely. That projection is a nondeterministic byte automaton, and
determinizing it by the subset construction gives a deterministic one whose state is the set of marking configurations
still alive after the bytes read. Non-membership is then decidable position by position, since the input read so far is
outside the side's domain exactly when no configuration in the current subset is accepting at the current end of input.
Search the product of the two subset automata for a state accepting on exactly one side. The state space is finite,
being pairs of subsets, so the search terminates, and parent pointers reconstruct the differing input.
\end{proof}

The auditor decides two predicates and returns a differing input whenever either holds. The first is boundary
disagreement on inputs both sets tokenize completely, from the product of the two armed-run verifiers, and the second is
domain difference, an input that exactly one side tokenizes at all, from the product of their determinized byte
projections.

The two bear on security, since request smuggling, polyglot inputs, and log mangling arise when two implementations read
one stream differently. At the lexical layer the auditor either proves a pair of token sets free of both differentials
or produces the diverging input. Neither predicate observes token identity where the boundaries coincide, or priority
ties among them, and that comparison is outside this instrument.

% machine-check: campaign_inventory.py, the differential emissions of
% data/campaign-splitting-stats.txt, witness re-verified against the reference scans;
% differential_regex.py, the domain decision beside the boundary decision, 42 ordered pairs each
% against brute force with every witness re-verified.
The auditor also decides whether the recovery campaign's change of whitespace treatment alters segmentation. The
conventional and split-friendly grammars of Section~\ref{sec:measurements} differ only in their whitespace treatment,
and the auditor synthesizes a shortest disagreement, the differing input \code{SN}, a space then a newline, one
whitespace token under the conventional row with boundary set $\{0\}$ and two tokens under the split-friendly row with
boundary set $\{0, 1\}$. The conventional row against itself is proved differential-free over every input, a guarantee
that no test on a finite corpus can establish.

% machine-checks: finitely_verifiable.py, generic_full_equivalence_sweep (56 ordered literal pairs; domain verdicts
% agree with enumeration through length seven; returned domain and boundary witnesses re-verified by independent
% greedy scans); differential_regex.py (regular sets, 42 ordered pairs including the unbounded-delay shape).

Table~\ref{tab:cost} gives upper bounds on the state spaces these decisions search, each derived from the proof that
introduces its construction, and makes no claim about running time. The literal product of Theorem~\ref{thm:cutoff}
carries a pending state, the occurrence match's progress through the window, and the violation phase, and its tracker
offset ranges to $L$. Over $T = \{a^{L}\}$, the search on the window $(a, 0)$ discovers $1 + L(L + 1)/2$ product states
before refuting it, for every $L \ge 2$, including 37 at $L = 8$ and 8{,}257 at $L = 128$. The count is quadratic in $L$
while $|P(T)| \cdot h$ is $L$, so $O(|P(T)| \cdot h)$ does not bound the search. The literal route also admits plain
enumeration to length $N = h + 2L - 2$, which the counterexample of Section~\ref{sec:deciding} shows is attained.

The armed-run product of Theorem~\ref{thm:offline} and Corollary~\ref{cor:inventory} carries a verifier state and the
same occurrence progress, once per pair. The inventory rows separate three quantities. The count of certification
decisions is exact, one for each (window, origin) pair of width at most $H$. The peak bound is the single-decision bound
evaluated at $h = H$. The aggregate bound sums the bounds over all pairs, using $h = j$ for each width-$j$ window.

The credit product of Theorem~\ref{thm:density} adds the byte buffer and its flag vector. Its buffer-and-flag factor is
written as a sum so that it equals one for an empty alphabet, where the product keeps its start state. The quantity $M$
is defined only for a nonempty inventory, so the empty inventory has its own row. That case is decided by testing
whether the token-language union is empty, before any credit product is built. If it is, only the empty input tokenizes
and the span is $0$. Otherwise tokenizable inputs grow without bound, and so does the span. Among the table's
constructions, only domain difference, the second predicate of Theorem~\ref{thm:diff}, uses subset determinization,
because it must establish non-membership on one side.

\begin{table}[!htbp]
\centering
\begin{tabular}{@{}L{0.3\textwidth}L{0.27\textwidth}L{0.36\textwidth}@{}}
\toprule
Decision & Construction & Search-space bound \\
\midrule
Window certification, literal & pending automaton, occurrence, phase, tracker & $O(|P(T)| \cdot h \cdot L)$ \\
\addlinespace
Window certification, regular & armed-run product, per $(W, o)$ & $O(|V| \cdot h)$ \\
\addlinespace
Anchor inventory through budget $H$, literal & $\sum_{j=1}^{H} j\,|\Sigma|^{j}$ certification decisions, one per
pair & peak $O(|P(T)| \cdot H \cdot L)$, aggregate $O(|P(T)| \cdot L \sum_{j=1}^{H} j^{2} |\Sigma|^{j})$ \\
\addlinespace
Anchor inventory through budget $H$, regular & the same count of decisions, each by the regular construction & peak
$O(|V| \cdot H)$, aggregate $O(|V| \sum_{j=1}^{H} j^{2} |\Sigma|^{j})$ \\
\addlinespace
Anchor-free span, nonempty inventory & credit product & $O(|V| \sum_{k=0}^{M-1} (2|\Sigma|)^{k})$ \\
\addlinespace
Anchor-free span, empty inventory & emptiness of the token-language union & the number of reachable states of the union
automaton \\
\addlinespace
Boundary differential & two verifiers in product, one flag & $O(|V_{1}| \cdot |V_{2}|)$ \\
\addlinespace
Domain difference & subset determinization of each projection & $O(2^{|V_{1}|} \cdot 2^{|V_{2}|})$ \\
\bottomrule
\end{tabular}
\caption{Upper bounds on the state spaces searched by the decisions of Sections~\ref{sec:deciding} and~\ref{sec:diff},
read off their constructions. For a literal token set, $P(T)$ is the set of proper token prefixes, the empty prefix
included, so that $|P(T)|$ counts the pending automaton's live states, and $L$ is the maximum token length. For a
regular family, $|V|$ is the state count of an armed-run verifier $V$, and $V_{1}$ and $V_{2}$ are the verifiers of the
two compared families. The alphabet is $\Sigma$, $h$ is the window length, $H$ the declared window budget and $M$ the
maximum window length in a nonempty inventory.}
\label{tab:cost}
\end{table}

\section{Measurements}\label{sec:measurements}
% Every number in this section is emitted by a named program into its file under data/ or derived by
% scripts/check-paper4.sh from those emissions; the tables input directly from the fragments the
% programs write, and the gate holds every prose number to its source the way the companion's gates
% hold its figures.

This section measures, on named material, how many anchors certificates supply, how a frozen inventory transfers, how
far an edit's damage reaches, and what the four campaign grammars offer. It applies the decisions of
Section~\ref{sec:deciding}, the splitting and edit theorems of Section~\ref{sec:splitting}, and the auditor of
Section~\ref{sec:diff}.

Anchor supply is defined for every token set. The vocabulary study needs literal token sets with hundreds to thousands
of strings sharing prefixes. Both byte-pair vocabularies named below meet this requirement, unlike the lexer grammars in
this study and in the recovery companion's measurement campaign. The locally trained vocabulary also comes with its
training text and an adjacent held-out slice. The campaign's four C-like grammars, measured at the end of this section,
are another kind of token set to which the decisions apply.

The first study reads two vocabularies over a held-out byte slice of the simdjson benchmark fixture. The program reads
the fixture through a one-to-one mapping from bytes to characters, so every quantity is a byte quantity and every slice,
window and substitution acts on bytes, never on decoded code points. That fixture is \code{twitter.json}, taken verbatim
at $631{,}515$ bytes and byte-identical to \code{jsonexamples/\allowbreak twitter.json} at that project's release
v3.10.1~\citep{simdjsondata}. It is a published benchmark fixture and was not collected for this work. The artifact
fetches it from that release, holds it to its digest and does not redistribute it, and it is used here only as a byte
stream. The JSON was not parsed or semantically inspected. The vocabularies are local-384, a byte-pair vocabulary of 384
merges trained here on the fixture's first 64 KiB over the full 256-symbol fallback base, and GPT2-prefix-4k, the first
4{,}096 merges of the GPT-2 byte-level vocabulary over its full byte alphabet, nothing pruned to the sample.

The vocabularies' token strings come from BPE merge tables, which native BPE applies by rank, while these studies ignore
the ranks and scan the strings by maximal munch, the scan the theorems speak about. Every figure below describes that
policy over those strings and measures nothing of BPE's own segmentation. The local training merges the most frequent
adjacent pair at each step, resolving ties by the shorter concatenation and then by the lexicographically greatest pair.
This is not the first-occurrence rule, and the difference identifies the vocabulary, since $251$ of the $384$ steps had
a most-frequent tie and $185$ of those would have merged a different pair under first occurrence.

The evaluation slice is the byte range $[65536, 81920)$, exactly 16 KiB, adjacent to the training bytes. For each
vocabulary, the inventory contains all certified bytes and every certified (window, origin) pair of widths two through
four whose window occurs in that slice. Every alphabet byte is decided by the interior-byte lemma and cross-checked
against the window decider, and every candidate window is decided at every origin. Because the candidates are drawn from
the slice the supply is then measured on, these rows report a ceiling measured after the fact, exact through the budget
on this slice, and not the supply an inventory frozen before the input would deliver. The two are separated below.
Table~\ref{tab:supply} gives the supply, counting distinct internal anchor positions, each once however many
certificates witness it, together with the gaps between consecutive anchors and the snap, the distance an ideal
eight-way cut moves to reach its nearest anchor.

\begin{table}[!htbp]
\centering
% Generated by splitting_measurements.py; edit the program, never this file.
\begin{tabular}{@{}lrrrrr@{}}
\toprule
Vocabulary & anchors/KiB & gap p50 & gap p90 & gap max & snap max (8) \\
\midrule
local-384 & 242.1 & 1 & 14 & 61 & 9 \\
GPT2-prefix-4k & 651.2 & 1 & 2 & 11 & 1 \\
\bottomrule
\end{tabular}

\caption{Certified-anchor supply of the two vocabularies on the evaluation slice, emitted by
\code{splitting\_measurements.py}. Gaps and snap distances are in bytes, and p50 and p90 denote the median and the
ninetieth percentile. The snap max (8) column gives the largest distance any of the seven ideal cuts of an eight-way
split moves to reach its nearest anchor. Gap statistics here are distances between consecutive anchor positions, so an
interior gap of $d$ bytes holds $d - 1$ anchor-free positions, and the runs before the first anchor and after the last
are excluded. The longest anchor-free stretch reported in the text also counts those runs.}
\label{tab:supply}
\end{table}
\FloatBarrier

Both vocabularies supply many anchors, and the external one supplies more. The trained vocabulary yields 242.1 anchors
per KiB and the external one 651.2, a ratio of about 2.7 on this slice. The trained vocabulary also has the longer tail,
its longest anchor-free stretch running to 60 bytes against the external one's 10 and its eight-way cuts snapping within
9 bytes against 1, while the two agree at the median gap of one byte. Position by position, 2{,}937 positions are
anchors under both vocabularies, 937 under the trained one alone, 7{,}482 under the external one alone, and 5{,}027
under neither, of the slice's 16{,}383 interior positions, a Jaccard index of 0.259 between the two anchor sets. This
difference in supply is consistent with frequent evaluation substrings having become token interiors under training, but
the two vocabularies differ in provenance, merge depth, training corpus, and construction at once, so the comparison
cannot identify that mechanism.

Merge depth, one of those four differences, can be varied alone. The local depth series holds the training corpus, the
evaluation slice, the fallback base, the scan convention, and the window budget fixed, and changes only the number of
merges. Supply falls at every step, 508.6 anchors per KiB at 48 merges, 420.6 at 96, 329.4 at 192, and 242.1 at 384, a
factor of $2.10$ across that range, while the longest anchor-free stretch grows from 21 bytes to 60. GPT2-prefix-4k lies
an order of magnitude deeper, at 4{,}096 merges, and differs from all four checkpoints in provenance, training data, and
construction, so the series cannot measure a depth contribution to the cross-vocabulary gap. It does rule out an account
in which supply depends on depth alone and falls with it, since the deeper vocabulary supplies 2.7 times as many
anchors; separating depth from the other axes would need matched depths across several training domains.

A slice-wide average hides variation between blocks. Across the sixteen kibibyte blocks the slice divides into, the
trained vocabulary supplies between 131 and 448 anchors per block with a median of 242, and the external one between 537
and 756 with a median of 688. Each anchor is detected once on the full slice and binned by its position. Percentiles
here and throughout this section take, from $n$ observations sorted in ascending order and for the fraction $\alpha$,
the entry at zero-based index $\min(n - 1,\, \lfloor \alpha n \rfloor)$. At $\alpha = 1/2$ this is the upper median when
$n$ is even. The ranges do not meet, which supports the comparison above, and each is wide enough that a single
aggregate would hide it.

The inventory is exhaustive through the declared window budget. Every window of widths two through four occurring
anywhere in the slice is decided at every origin, 20{,}690 decisions over 6{,}307 windows for each vocabulary, and the
budget curve reports the supply gained at each width. Through $H = 1$, the certified bytes alone, the trained vocabulary
supplies 52.2 anchors per KiB and the external one 379.8, while widths two, three, and four carry the trained one to
208.9, 230.7, and 242.1, and the external one to 608.8, 620.9, and 651.2. Each width's gain can be set against the
number of decisions it took, which is the cost of widening the budget. Width two carries most of the gain, 82.5 and 84.4
percent of what widths two to four add beyond the certified bytes, for the trained and the external vocabulary in turn,
and it gains 2{,}507 anchor positions on the trained vocabulary and 3{,}665 on the external one for $2{,}316$ decisions.
Width three gains 349 and 193 for $6{,}666$ decisions, and width four gains 183 and 485 for $11{,}708$. The curve has
not flattened at the budget, and on the external GPT-2 vocabulary width four gains more than width three did, so what
lies past $H = 4$ is not estimated here.

The rows above are, then, a ceiling, the supply an inventory chosen with the input already in hand could certify through
the budget. A deployment fixes its inventory first and meets the input afterwards, and the difference is measurable. The
frozen inventory is a hybrid of every certified byte, which the interior-byte lemma settles from the token set alone,
together with every certified window of width two to four that occurs in the adjacent 16 KiB, the byte range
$[81920, 98304)$, which is disjoint from both the training bytes and the evaluation slice.

That inventory is applied out of sample, on the evaluation slice. There the trained vocabulary supplies 209.4 anchors
per KiB against its ceiling of 242.1, which is $86.5$ percent of it, and the external one 608.6 anchors per KiB against
651.2, or $93.5$ percent. On this one adjacent-slice transfer, the frozen inventory costs $13.5$ and $6.5$ percent of
the aggregate supply. The external vocabulary loses the smaller share of aggregate supply, but its longest anchor-free
stretch grows from 10 bytes to 35, while the trained vocabulary's is unchanged at 60 bytes.

The transfer is also reported per block. For each of the sixteen kibibyte blocks of the evaluation slice it is the share
of the block's ceiling supply that the frozen inventory keeps, and Table~\ref{tab:frozen} gives the range and upper
median of those shares. These are held-out figures from a single calibration split, so a deployment that needs a
worst-case bound should use the grammar-global anchor-free span when it is finite.

\begin{table}[!htbp]
\centering
% Generated by frozen_inventory.py; edit the program, never this file.
\begin{tabular}{@{}lrrr@{}}
\toprule
 & \multicolumn{3}{c}{share of the block's ceiling supply retained, percent} \\
\cmidrule(lr){2-4}
Vocabulary & lowest block & upper median & highest block \\
\midrule
local-384 & 71.0 & 88.5 & 100.0 \\
GPT2-prefix-4k & 85.0 & 94.0 & 100.0 \\
\bottomrule
\end{tabular}

\caption{Share of each evaluation block's ceiling supply retained by the frozen hybrid inventory, emitted by
\code{frozen\_inventory.py}. The inventory's windows of widths two to four were calibrated on the disjoint slice, and
the upper median over the sixteen blocks follows the zero-based percentile rule stated in the text.}
\label{tab:frozen}
\end{table}

The exactness matters because a sampled inventory of the same slice misses many anchors, especially for local-384. The
candidate set is sixty frequent slice n-grams, the twenty most frequent of each width. Every occurrence of every n-gram
of that width is counted, the final one included. Distinct n-grams are ranked by that count, and frequency ties are
broken by first occurrence in the slice. The rule has to be stated because the rank-20 selection boundary is tied. At
widths two, three and four, five candidates share the twentieth frequency and tie for the last two, one and four places
respectively, so ``the twenty most frequent'' names a family of candidate sets and not one. The trained vocabulary, over
all 250 selections consistent with those ties, certifies exactly one pair whichever is chosen while the external one
certifies between 35 and 38, and the stated tie rule selects 38.

Across all 250 selections the anchors per KiB the selected pairs yield are invariant, whether the pairs are counted
alone or together with the certified bytes, each position counted once. Within this family of selections, then, the tie
rule moves only the external vocabulary's certified-pair count, and the two per-KiB figures below do not depend on it.

Those sixty candidates yield a single certified pair for the trained vocabulary and 5.6 anchors per KiB of their own,
against the 189.9 the exhaustive inventory's wider component adds beyond the certified bytes on the same slice. On the
external GPT-2 vocabulary the same sixty candidates do much better, yielding 38 certified pairs and 370.9 anchors per
KiB. For each vocabulary the sampled inventory keeps the exhaustive one's exact certified bytes, and its anchors are the
union of theirs with the sampled-pair anchors, each position counted once. The trained vocabulary's sampled inventory
comes to 57.8 anchors per KiB against the exhaustive 242.1, 23.9 percent of it, and almost all of that is the 52.2 the
certified bytes supply on their own. The external vocabulary's comes to 427.4 against 651.2, 65.6 percent of it; its
certified bytes supply 379.8 alone and its sampled pairs 370.9, the two sets of positions overlapping heavily. The
emitted statistics label this union ``sampled wider half added to the exact singleton half''.

The pattern fits the token-interior reading above, with the same limit. A failed certificate for an n-gram means that,
in some completely tokenizable context, the token covering the n-gram's final byte starts somewhere other than the
proposed origin. It does not establish that the whole n-gram lies inside a token of the trained vocabulary. Nor do the
two vocabularies differ in alphabet coverage, since both carry all 256 fallback bytes, so every byte string is
tokenizable under either. They differ in the token inventory above that base, built by merges of different provenance,
and a certification verdict depends on that inventory alone.

Incremental damage is measured over 200 random byte substitutions per vocabulary. The draws come from seed 20260826,
reset per vocabulary, with positions uniform over the slice and replacement bytes uniform over the 256-byte alphabet. A
draw that repeats the original byte or leaves the input untokenizable is redrawn. Of 201 attempted draws 200 were
accepted; the one rejection repeated the original byte, and none was rejected for failing to tokenize, since the
fallback base makes every byte string tokenizable.

In every trial an anchor is eligible when at least one of its witnessing occurrences lies wholly in the unchanged
suffix, and the right bounding anchor $q$ is the least eligible anchor the declared $H \le 4$ inventory supplies. It is
not the least such anchor over all widths, because a certificate wider than the budget can witness an earlier eligible
anchor, and a probe one width past the budget finds one in $2$ of the $200$ trials on the trained vocabulary and $1$ on
the external one. That one-width probe is incomplete and does not determine the global least, though every qualifying
anchor it finds is a valid, possibly nonminimal, confinement bound.

For a fixed trial the global least is computable, by testing every substring of the unchanged suffix at every origin,
since every witnessing occurrence lies wholly in that suffix, and this paper does not run that search. The budget makes
the slack inventory-relative, so the slack reported is an upper bound on the global slack.

An accepted edit either moves some boundary or moves none. For each vocabulary, Table~\ref{tab:edit} reports how often
an accepted substitution moves nothing at all, the damage conditional on it moving something, the unconditional damage
with hull zero assigned to the trials that moved nothing, and the number of boundaries moved. Whenever a boundary moves,
the hull, the inclusive span from the first moved boundary to the last, lies inside Theorem~\ref{thm:edit}'s interval.
Means in the table are printed to three decimals; the unconditional means, whose denominator is 200, are exact at that
precision.

The two vocabularies are paired within each of the 200 trials. A trial's admissible width is the number of integer
positions strictly between $\max(0,\, p - L + 1)$ and $q$, the positions that the length-bounded interval of
Theorem~\ref{thm:edit} with $a = 0$ leaves open to a moved boundary. The paired difference, trained minus external, has
median 3 bytes and ninetieth percentile 13 on the hull, and 11 and 25 on the admissible width. A positive median does
not make the trained hull larger in every trial, since the hull difference is positive in 144 of the 200 trials, zero in
38, and negative in 18.

Whether an edit moves anything at all is paired the same way. The same substitution moves a boundary under both
vocabularies in 106 of the 200 trials, under the trained one alone in 63, under the external one alone in 12, and under
neither in 19.

\begin{table}[!htbp]
\centering
\begin{tabular}{@{}lrr@{}}
\toprule
200 accepted substitutions per vocabulary & local-384 & GPT2-prefix-4k \\
\midrule
trials moving no boundary & 31 & 82 \\
trials moving a boundary & 169 & 118 \\
\midrule
conditional hull median & 6 & 1 \\
conditional hull ninetieth percentile & 13 & 2 \\
conditional hull maximum & 20 & 7 \\
conditional hull mean & 6.852 & 1.508 \\
\midrule
unconditional hull median & 4 & 1 \\
unconditional hull ninetieth percentile & 13 & 2 \\
unconditional hull maximum & 20 & 7 \\
unconditional hull mean & 5.790 & 0.890 \\
\midrule
boundaries moved, mean conditional on a change & 3.882 & 1.415 \\
boundaries moved, mean per accepted trial & 3.280 & 0.835 \\
\bottomrule
\end{tabular}
\caption{Incremental damage from one edit seed. The program \code{splitting\_measurements.py} emits every figure into
\code{data/splitting-stats.txt}. The hull is the inclusive span, in bytes, from the first moved boundary to the last. A
trial that moves nothing has no such span and is recorded with hull zero, which is a convention and not a measurement.
The unconditional median, ninetieth percentile, maximum and mean include those trials at hull zero.}
\label{tab:edit}
\end{table}

How much room the interval leaves beyond the disturbed span is measured by two statistics, which answer different
questions and coincide only when the moved boundaries fill their own hull. Both start from the admissible width. The
\emph{hull slack} subtracts the hull's inclusive width and measures how much wider the theorem's containment interval is
than the span the edit actually disturbed. The \emph{unoccupied count} subtracts instead the number of boundaries that
moved and counts admissible positions no moved boundary sits on.

Both are inventory-relative, so both are reported twice, once against the evaluation-tailored ceiling inventory of the
rows above and once against the deployable frozen inventory of the transfer above. The frozen inventory's anchors on the
evaluation slice are a subset of the ceiling's, so its admissible width in a trial is never smaller.
Table~\ref{tab:slack} gives both statistics under both inventories, and any agreement between them comes from these
trials and not from the definitions.

\begin{table}[!htbp]
\centering
% Generated by splitting_measurements.py; edit the program, never this file.
\begin{tabular}{@{}llrrrr@{}}
\toprule
 & & \multicolumn{2}{c}{ceiling inventory} & \multicolumn{2}{c}{frozen inventory} \\
\cmidrule(lr){3-4}\cmidrule(lr){5-6}
Vocabulary & Statistic & p50 & p90 & p50 & p90 \\
\midrule
local-384 & hull slack & 40 & 50 & 40 & 50 \\
local-384 & unoccupied count & 40 & 54 & 41 & 54 \\
GPT2-prefix-4k & hull slack & 31 & 34 & 31 & 36 \\
GPT2-prefix-4k & unoccupied count & 31 & 35 & 31 & 36 \\
\bottomrule
\end{tabular}

\caption{Median (p50) and ninetieth percentile (p90) of the hull slack and the unoccupied count under both inventories,
in admissible boundary positions. Each measures how much of Theorem~\ref{thm:edit}'s containment interval an edit left
unused, and \code{splitting\_measurements.py} emits both. The ceiling inventory is computed exhaustively on the
evaluation slice, and the frozen hybrid inventory combines all certified bytes with the certified windows of widths two
to four calibrated on the disjoint adjacent slice and applied out of sample.}
\label{tab:slack}
\end{table}

The synchronization distance of Corollary~\ref{cor:gap} is unbounded for both subword vocabularies, since each admits a
cycle of non-certified bytes in its pending automaton, while the designed UTF-8 shape has distance three. The delimiter
criteria of Theorem~\ref{thm:delimiter} refuse the newline byte for both vocabularies in both directions, since in each
vocabulary newline occurs in some token at a noninitial position and in some token at a nonfinal one, so cutting at or
after every newline is unsound for either vocabulary. Two depths of the local family also segment differently, since the
48-merge and 384-merge local vocabularies agreed on none of 64 held-out slices.

% Every number in this passage is emitted by campaign_inventory.py into
% data/campaign-splitting-stats.txt, the table input directly from the fragment that program
% writes, and each grammar abstraction is stated beside the program with its class
% abstraction checked by execution against a byte-level scan of its own archived corpus. The
% occurring-or-vacuous split of each row's certificates is decided by occurring_windows.py, the
% violation search without its origin bit, validated against bounded enumeration and by its own
% witnesses; campaign_inventory.py rescans every occurring witness, checks that no vacuous window
% occurs in its row's corpus, and emits the split per row and over the four rows.
The instruments are also applied to the four C-like grammars of the recovery companion's measurement campaign, each on
its whole archived corpus. Each grammar is represented by byte classes stated beside the program that measures it. The
program checks by execution that the byte-level greedy scan and the scan of the abstracted classes have the same token
starts throughout the grammar's archived half-mebibyte corpus, including all 149{,}842 token starts of the conventional
row. Every row's anchor-free span is unbounded, which rules out a grammar-global gap guarantee, so the supply figures
below are corpus measurements. On every row the program cuts the whole archived corpus at every certified anchor, and
the chunked scans are byte-identical to the sequential stream. Table~\ref{tab:campaign} gives the four rows.

\code{campaign\_inventory.py} also checks Theorem~\ref{thm:edit} on edits to the archived corpora. On each of the
conventional and split-friendly corpora, sixty cross-class single-byte substitutions are applied to a slice ending at a
line boundary. They are retained from eighty-eight draws per row, the twenty-eight whose edited slice no longer
tokenized being redrawn. On each attempted draw the position is uniform on the slice and the target class uniform among
the other classes, with one fixed representative byte per class. Every disturbed boundary moves strictly between the
flanking shared-context anchors, with a boundary-difference hull median of one byte and an observed maximum of two under
the declared seed. Twelve same-class substitutions per row move no boundary at all, as the class abstraction predicts.
They were retained from fourteen position draws, redrawing the two whose byte was already its class's sole
representative.

\begin{table}[!htbp]
\centering
% Generated by campaign_inventory.py; edit the program, never this file.
\begin{tabular}{@{}lrrrrr@{}}
\toprule
Campaign row & anchors/KiB & gap p50 & gap p90 & gap max & snap max (8) \\
\midrule
conventional & 25.7 & 43 & 69 & 104 & 15 \\
split-friendly & 41.0 & 21 & 64 & 104 & 15 \\
block-comment & 6.7 & 121 & 304 & 965 & 96 \\
bare & 324.4 & 2 & 6 & 10 & 3 \\
\bottomrule
\end{tabular}

\caption{Certified-anchor supply of the recovery campaign's four C-like grammars on their archived corpora, emitted by
\code{campaign\_inventory.py}. Gaps and snap distances are in bytes. Gap statistics here include the endpoints, the run
from position zero to the first anchor and the run from the last anchor to the end of the corpus, so they are not
comparable row for row with Table~\ref{tab:supply}, whose gaps exclude both.}
\label{tab:campaign}
\end{table}
\FloatBarrier

The four rows show how grammar design affects anchor supply, and one pair of them is a controlled comparison. The
conventional and split-friendly rows differ in the grammar alone over the same bytes, while the block-comment and bare
rows are case studies of their own grammars and corpora.

The program's class letters are $L$ a letter or underscore, $D$ a digit, $S$ a space or tab, $N$ the newline, $P$ a
punctuator, $Q$ the quote, $C$ the slash, $A$ the star, $O$ any other operator byte, and $X$ every byte in none of those
sets. Each row distinguishes only the classes its tokens tell apart and folds the rest into $O$ or $X$. A class pair
below is a (window, origin) pair over a row's class alphabet.

The conventional row has identifiers, numbers, operators, punctuation, whitespace, string literals, and line comments,
and holds 15 certified class pairs at windows up to length two, none of them a single class at origin zero. It supplies
25.7 anchors per KiB with gap percentiles 43 and 69 and maximum 104. With 13{,}171 anchors cutting the archive, its
eight-way splits snap within 15 bytes and sixty-four-way within 38. The split-friendly row runs on the same corpus bytes
with the newline as a separate token. The newline class itself certifies at origin zero, the supply rises to 41.0
anchors per KiB and the median gap falls from 43 to 21.

The block-comment row supplies the fewest anchors. No window of two classes certifies at all, and of the counterexamples
the decision procedure synthesizes for those refutations, the longest is a six-byte whole block comment. A closing
delimiter fixes a boundary only with one preceding class symbol, since in \code{/*/} the opener's star followed by an
interior slash has the same two classes as a close. Of the 81 close-shaped four-class windows, 65 certify at origin
three, 57 of them with a window that occurs in some completely tokenizable class string and 8 vacuously, as explained
below. The supply is 6.7 anchors per KiB with a maximum gap of 965 bytes. At the opposite extreme, the bare row folds
all operators into one certified class and all punctuators into another, supplying 324.4 anchors per KiB with a maximum
gap of 10 bytes.

\paragraph{Occurring and vacuous certificates.} Corollary~\ref{cor:existence} separates two readings of a certificate,
and the vocabulary and campaign inventories of this section report different ones. The vocabulary candidate inventories
enumerate windows occurring in the slice, so every certificate of Table~\ref{tab:supply} occurs. The campaign
inventories enumerate the whole declared class alphabet, so their counts are unrestricted, and a pair whose window
occurs in no completely tokenizable class string is certified vacuously and anchors nothing.
\code{occurring\_windows.py} decides, for every certified pair, whether its window occurs in some completely tokenizable
class string, and \code{campaign\_inventory.py} reports the split. Of the 156 certified pairs over the four rows, 117
occur and 39 are vacuous. The block-comment row contributes the 65 pairs of its selected close-shaped width-four
extension, since its inventory through width two is empty, and the other three rows contribute their inventories through
width two.

On the conventional row 13 of the 15 pairs occur, and on the split-friendly row 17 of the 19. In both rows the two
vacuous pairs are the window $NX$ at both origins, since $X$ occurs only inside a string literal or a line comment,
neither of which contains a newline, and no token begins with $X$.

On the block-comment row 57 of the 65 occur, and the eight vacuous ones are $xACX$ at origin three for every class $x$
but the slash. They are vacuous because $X$ occurs only inside a block comment, a star followed by a slash past the
opener closes one, and no token begins with $X$. When $x$ is the slash class $C$, the window $CACX$ occurs as an opener
followed by an interior slash, and the pair is refuted.

On the bare row 30 of the 57 occur, and the 27 vacuous ones are the pairs whose window uses $X$, a class of the row's
alphabet that no token of the row contains. This distinction changes none of the anchor-supply, gap, snapping,
edit-bound or split-equality results. A window that occurs in no tokenizable input generates no corpus anchor, and the
program checks that none of the windows of these 39 vacuous pairs occurs in its row's corpus.

\paragraph{Scope of the measurements.} The vocabulary studies use one corpus and one 16~KiB slice of it, with an
adjacent slice for the frozen transfer. They read two primary vocabularies under maximal munch, their native merge ranks
set aside, together with a four-depth local series over the same corpus and slice. The main inventories use a window
budget of $H \le 4$. The edit studies each use one fixed seed, 20260826 for the vocabulary trials and 20260827 for the
two campaign edit rows. The campaign rows run on generated corpora, and the timings come from one machine and are
illustrative. The theorems hold for every token set and input satisfying their stated hypotheses. The figures in this
section are what those theorems produced on that material, and nothing here generalizes them to another corpus,
vocabulary, or budget.

\FloatBarrier
\section{The artifact: the theorem as a scanner}\label{sec:artifact}

% machine-check: run-artifact.sh over the munch probe tools/probes/src/parallel_scan.cpp, pinned in
% that script by repository, commit, and source digest; every number here is read from
% data/artifact-stats.txt, timing lines recorded rather than enforced.
The artifact is a threaded scanner probe, \code{parallel\_scan.cpp} in the \code{munch} library, that cuts a corpus at
certified anchors, scans the chunks in parallel, joins the results, and asserts that the boundary stream equals the
sequential one. Its anchor inventory is derived from the campaign grammar alone and computed before any input is read.
The probe is pinned by repository, commit and source digest, as Appendix~\ref{app:provenance} records. Timing lines
depend on the machine, so they are recorded and never enforced, while every other line is stable and compared on
regeneration.

On the archived half-mebibyte corpus, with 149{,}842 token starts, the snap maxima are 15, 26, 28, and 38 bytes at 8,
16, 32, and 64 workers, and the chunked stream is byte-identical in every configuration. On sixteen concatenated copies,
8{,}388{,}608 bytes and 2{,}397{,}472 token starts with 210{,}751 anchors, the snap maxima are 0, 0, and 6 bytes at 8,
16, and 32 workers, and the equality holds again. That larger input repeats the same corpus sixteen times, so it tests
scale without adding workload variety. The two zero snaps are a fact about its layout and say nothing about snapping at
scale, since the fifteen copy joins are anchors, $210{,}751 = 16 \times 13{,}171 + 15$, and the eight- and
sixteen-worker grids fall on them.

Each timing is the best of five runs on one machine, and the environment file the runner writes also records five
whole-probe repetitions of the half-mebibyte configuration. The sixteen-copy figures have no such repetitions. On the
sixteen-copy input, the recorded sequential scan takes 35.06 milliseconds and the 32-worker split 5.70, a speedup of
6.16, while the half-mebibyte configuration reaches 3.06 at 8 workers over a two-millisecond scan, and its best speedup
across the five repetitions runs from 3.16 to 3.52, at 16 workers in three of them and at 8 in two. This pattern is
consistent with a fixed parallel overhead mattering more as the input shrinks, though nothing here isolates that
overhead.

The timings cover scanning and, for parallel runs, thread creation and joining, but they exclude anchor discovery, cut
selection and output concatenation. Sequential runs build a fresh boundary vector each repetition, while parallel runs
reuse each chunk's vector capacity across repetitions. The theorem guarantees the equality of the boundary streams,
while the snap distances are measured and depend on the input, the inventory, and the worker grid, and the timings also
depend on the machine. Appendix~\ref{app:provenance} records how to regenerate every emission and check the scanner's
provenance.

\section{Related work}\label{sec:related}

\paragraph{Code synchronization and resolving blocks.} The certificate has a classical special case, and the window
paper places it in that lineage \citep{nidhogg2026splitwindows}. Here we use the restricted convention of
\citet[p.~17-01]{lassez1976messages}, in which a synchronizing pair of a code $X$ is a pair $(x, y)$ of words of $X^{*}$
such that $uxyv \in X^{*}$ forces $ux \in X^{*}$ and $yv \in X^{*}$ for all $u$ and $v$, so an occurrence of $xy$ in a
word of $X^{*}$ fixes the cut between its halves in every context. \citet[p.~373]{berstel2010codes} state the general
convention, in which $x$ and $y$ may be any words over the alphabet. A code has finite synchronization delay when some
number of codewords on each side makes every such pair synchronizing \citep{restivo1975synchronization}. For a regular
code that bounded delay is decidable, through the syntactic monoid of $X^{*}$
\citep[Th\'eor\`eme~2]{lassez1976messages}, by a finite-state argument of the kind the products of
Section~\ref{sec:deciding} make for token sets that need not be codes. \citet{fici2025morphisms} study morphisms
synchronizing with finite delay, a related bounded-window recovery notion for the circular factors of morphic images.
For prefix codes the synchronizing words are themselves an algorithmic subject, with shortest ones for Huffman codes
\citep{biskup2009shortest} and for prefix codes and their decoders \citep{ryzhikov2018short}.

Over a prefix code read as a token set, the completely tokenizable inputs are the words of $X^{*}$, and the scan
computes their unique factorization. A window occurring in a word of $X^{*}$ is certified exactly when every occurrence
of it in every word of $X^{*}$ has a factorization boundary at the origin and its part from the origin on is a prefix of
a codeword. The first condition fixes the factorization cut, and the window's parts before and after it need not
themselves lie in $X^{*}$, as both halves of a pair in this restricted convention must. The second makes the codeword
starting at the cut cover the window's tail, since a codeword starting at the cut and ending inside the tail would be a
proper prefix of a codeword extending the tail, which prefix-freeness excludes.

In symbolic dynamics a resolving block is a block of a factor subshift all of whose admissible preimages agree at a
selected coordinate \citep{adler1983sliding, marcus1985sofic}. Certification resembles this, with the window as the
block and the boundary bits from its origin to its end as the values every preimage agrees on, although resolving blocks
are stated for factor maps of subshifts, without token priorities or a finite input.

The C-like grammars and the byte-pair vocabularies measured here differ from the code case in not being prefix codes.
Each C-like token set measured here contains identifiers that are proper prefixes of longer identifiers, so none is
prefix-free. A token set need not be a code at all, as $\{\code{a}, \code{ab}, \code{b}\}$ shows, since \code{ab} has
two factorizations over it, and maximal munch selects one without relying on unique decipherability. For such sets the
certificate is defined through the scan and decided by the constructions of Section~\ref{sec:deciding}, the literal
cutoff and the armed run, rather than on a code's automaton.

On the code side, bounded synchronization delay is decidable for a regular code, as recalled above, and whether a
strongly connected partial automaton, the decoder of a prefix code among them, is synchronizing is checked in polynomial
time \citep[Section~2]{ryzhikov2018short}. On the transducer side, Karttunen's construction below makes the certificate
of a proposed window decidable by standard closure. Beyond those, no prior work was found that poses, for prioritized
token languages under maximal munch, whether a window's every occurrence fixes the origin of the covering token, states
the existence decision of Corollary~\ref{cor:existence}, or gives the constructions of Section~\ref{sec:deciding}.

\paragraph{Computing the scan.} Maximal munch is the standard convention, and its algorithmics have an established
literature. \citet{reps1998maximal} shows that a scanner can perform it in time linear in the input. He contrasts it
with the textbook backtracking scanner, which is quadratic on token sets such as $\{a,\ a^{*}b\}$, whose scan may look
arbitrarily far past a final state before committing. \citet{li2025uniform} give algorithms for the uniform version of
the problem, in which both the grammar and the text are inputs, that run in time linear in the text for a fixed grammar.
For WordPiece's finite vocabulary, \citet{song2021wordpiece} scan in linear time by keeping the uncommitted suffix as a
node of the vocabulary trie. When the trie cannot continue, they emit that node's failure pops, the tokens precomputed
for a failed extension, and follow its failure link to the node representing the still-uncommitted suffix. The trie node
represents the pending tail of Lemma~\ref{lem:pending}, which WordPiece maintains during scanning and this paper
explores when deciding certification. Reps states two conventions this paper inherits, that a tie over the longest match
is broken by the earlier token definition and that no token language admits the empty string. That literature answers
how fast the scan runs, and the question here, which of its positions are safe to start from, is a separate one it does
not address.

\paragraph{Longest match as a finite-state relation.} \citet{karttunen1996directed} compiles left-to-right longest-match
replacement into finite-state transducers. He builds tokenizers with them by inserting an end-of-token mark after each
match, an encoding of segmentation as a marked string, which the markings of Theorem~\ref{thm:offline} take up in bit
form. \citet{kaplan1998tokenizing,kaplan2005tokenizing} tokenizes with a finite-state relation and emits output at pinch
points, positions of the text at hand where all tokenization paths reaching them can be extended in exactly the same
ways. Pinch points express convergence on one input's configurations, where a certificate asks that a fixed window fix
the covering token's origin in every context.

Two of Karttunen's observations bear directly on Section~\ref{sec:deciding}. He notes that without the left-to-right
longest-match constraint the tokenizing transducer would not produce deterministic output, since it must place an
end-of-token mark after a letter sequence exactly when the word is not part of some longer multiword token. He also
shows that the transducer for $a^{+}b$ is unambiguous yet cannot be sequentialized, because ``it is obviously impossible
for any finite-state device to accumulate an unbounded amount of delayed output''. The same unbounded-lookahead pattern
appears in Reps's quadratic example, and Section~\ref{sec:deciding} uses it in $T = \{\code{a}^{+}\code{b}, \code{a}\}$,
on which certification is decided offline while no finite-state scanner emitting boundaries irrevocably, left to right,
can exist.

Two consequences already follow from his construction. The correctly marked greedy streams of a regular token set form a
regular language, since a transducer marking maximal instances has a regular output range. Once that language is
effectively regular, deciding whether a given marking belongs to it is membership testing, so the existence of an
offline verifier is an inference from his work and is not claimed as a contribution here. The closure operations
Corollary~\ref{cor:existence} spells out then give the bare decidability of a proposed window's certificate and of
existence at unrestricted width, which is not claimed as a contribution here either.

Theorem~\ref{thm:offline} supplies an explicit armed-run automaton over marked bytes, which carries no token identities,
has size bounded in terms of the token set, and is built to be intersected with a window's badly covered occurrences.
That intersection turns a verifier into a certificate test, and it is why the paper states the construction and not only
the fact that some verifier exists. Karttunen constructs the transducer, and this paper contributes the
characterizations of the properties of the segmentations it produces, the existence decision with its shortest member,
and the constructions that return the counterexamples and differing inputs, not their decidability as such.

\paragraph{Equivalence of transducers, and what it decides.} The two predicates of Theorem~\ref{thm:diff} mirror the
decomposition \citet{veanes2012sft} state for symbolic finite transducers, where deciding $A \equiv B$ ``reduces to two
independent tasks'', domain equivalence and equality of outputs on the common domain. This decomposition is the general
reason a comparison of this shape is decidable, and it is stated of single-valued transducers over a decidable label
theory. The decomposition does not construct the objects being compared, and two maximal-munch policies are not given as
transducers. Karttunen's construction above represents each policy's segmentation map by a transducer, while this paper
uses Theorem~\ref{thm:offline}'s armed run, an acceptor of correctly marked scans that validates guessed marks and
confines them to the policy's own scan. The contribution here is the construction of those acceptors and the product
searches that return the differing inputs. Veanes et~al.'s own equivalence decision is described as nonconstructive,
with counterexample generation supplied by a separate procedure.

\paragraph{Incremental and checkpointed lexing.} Theorem~\ref{thm:edit} derives its edit bound from certificates
determined by the token set, whereas incremental lexers use state saved from earlier scans. Incremental lexing in the
style of \citet{wagner1997lexing} restarts from per-token scanner-state snapshots, exact for the text they were cached
against, in a versioned token stream with dynamically maintained lookahead dependencies. The confinement is then a
property of a recorded run and of the state carried across edits, and it is available wherever that state was kept. The
bound here requires no saved scanner state. A certified anchor whose witnessing occurrence lies in bytes the edit left
untouched is a boundary of both versions by its certificate alone, so the certificates and their occurrences surviving
around the edit determine the damage interval, even where no snapshot exists. Interior anchor endpoints are available
only where such occurrences supply anchors. Certified anchors and saved scanner states can be used together, and
Section~\ref{sec:measurements} measures how much of the theorem's interval a real edit leaves unused, from which a
reader can judge the practical value of the certificate-only bound.

The author's parser library hopper~\citep{hopper} implements the restart for range edits, beyond the single-byte
substitution Theorem~\ref{thm:edit} is stated for. Its \code{Document} type searches backward for a certified anchor
whose evidence the edit left untouched and rescans from it to the first shared boundary past the edit, rescanning the
whole document when that bounded search finds no certified anchor. The hopper repository records its own measurements,
which this paper's gate does not check.

\paragraph{Cutting at a boundary, and deciding where one is.} \citet{chassot2026ziplex} verify a linear-time
maximal-munch lexer under the same two conventions this paper uses, non-nullable rules and priority by rule order. Its
proofs rest on helper lemmas about properties of lexing, among them ``the fact that removing a suffix of the input does
not affect the lexing of the prefix, provided that the cut occurs at token boundaries''. Their lemma, iterated, implies
Lemma~\ref{lem:cut}, and Theorem~\ref{thm:split} applies that lemma at certified anchors. Their closest result is a
conservative predicate $\mathit{sep}(t_1, t_2)$ on already known adjacent tokens, and pairwise satisfaction makes the
printed token sequence re-lex unchanged. That predicate starts from concrete adjacent token texts, while a certificate
is decided from a raw window over every tokenizable context.

\citet{lester2016flow} check the same hypothesis at a known join when concatenating analyzed code strings. They require
every lexer state reachable after the first string either to have emitted its token or to emit the same token on every
character that may begin the second, so that lexing the concatenation is concatenating the lexings. That check is made
per concatenation site, over the strings a string analysis admits for the two operands. The ZipLex lemma assumes a
boundary of a concrete input, while the join check establishes a sufficient boundary condition at one analyzed site.
This paper supplies a certificate stronger than a boundary at the origin, computed from the token set before any input
exists, which makes the window's origin a boundary at every occurrence in every completely tokenizable context. Neither
$\mathit{sep}$ nor the join check is the certificate.

The weaker predicate asks whether every occurrence of a window in every completely tokenizable input forces a boundary
at its origin, without fixing the token covering the window's final byte. In the window paper's third
version~\citep{nidhogg2026splitwindows}, which adds the boundary profile, the predicate says that the origin's gap is
forced in that profile, and that version decides it through this paper's Theorem~\ref{thm:offline}. This paper does not
decide it, the wider class of occurrence-based sound-cut rules being out of scope, as Section~\ref{sec:setting} says.
The introduction and Section~\ref{sec:setting} separate the two predicates, with $(\code{b}, 0)$ over
$\{\code{ab}, \code{b}\}$ refuted by the input \code{ab}, whose position 1 lies inside a token, and $(\code{ab}, 0)$
over $\{\code{a}, \code{b}\}$ refuted although its origin is a boundary in every context.

\paragraph{Parallel scanning.} \citet{mytkowicz2014dataparallel} parallelize finite-state machines, tokenization among
their applications, by enumerating transitions from every possible start state on each chunk and afterwards selecting
the run that began in the correct one. The method is exact, and its cost is the enumeration, naively an overhead
proportional to the state count, which their convergence optimization reduces to the states still active. The artifact
of Section~\ref{sec:artifact} is exact another way, cutting where a theorem already fixes the start state, so it needs
neither the enumeration nor a reconciliation pass, but it can cut only where the grammar's certificates occur. The two
are complementary. Certified splitting is available wherever the chosen inventory supplies anchors, and
Section~\ref{sec:measurements} measures that supply on the material it studies.

Other parallel-lexing methods incur different costs. Speculation predicts an entry state, validates, and re-executes the
chunk when the prediction fails \citep{prabhu2010speculative}. In the patent literature a partition scanner cuts a
stream at likely token boundaries, scans the chunks in parallel, and discards the later chunks' results and restarts at
the actual boundary when a guess is refuted \citep{vonpraun2011speculative}. An XML preparser carries the runs of a
chunk from every possible entry state in one meta-DFA state, and then keeps the run that the preceding chunk's end state
selects \citep{pan2007metadfa, chiu2014metadfa}. Bounded alternative-state scans follow every run that a chunk's
possible entry states admit and reconcile at the join, with the syntax restricted where the alternatives would not stay
bounded \citep{barenghi2015parallel}. Prescanning requires a pass over the input \citep{li2021plex}, and
\citet{yang1996mealy} treats longest-match lexing with Mealy machines, with a data-parallel algorithm among the results.

\citet{shao2026lopt} claim exact equality with sequential tokenization for BPE and WordPiece tokenizers by another
route, overlapping the chunks and merging adjacent ones at overlap tokens that sit at the same character positions in
both. Their Theorem 3.1 states a sufficient condition, a position-matched span strictly longer than the vocabulary's
longest token, and the published algorithm restarts with a doubled chunk only when it finds no overlap token at all, a
check that does not enforce that length condition. Shao et al. merge the chunks at their overlap after scanning, while
the state enumeration of Mytkowicz et al. does its extra work during the scan. Certified splitting only has to locate an
anchor near each cut, a cost the artifact's timings exclude, and it can cut only where the token set supplies an anchor.

\citet{zhang2026toktier} use a related matched-overlap idea for text appended to a session, under a different acceptance
rule. They match fresh tokens against cached token records and accept a splice only through a long enough equal run that
contains a family-specific stable boundary. For their byte-level BPE families that boundary is a character-class
transition from a per-family synchronizing set, which they report as resetting the pre-tokenizer under every left
context. A failed check widens the window and finally re-tokenizes the request. Their appendix calls the condition
checked at the junction a splice certificate, proves that a splice satisfying it is lossless, and summarizes how each
tokenizer family satisfies it, with the full proofs in a separate document shipped with the artifact. Among the serving
systems surveyed, those per-family acceptance predicates are the closest analogs of a certificate among checks made at
run time. They are established for a fixed pipeline, whereas the inventories here are decided from the token set.

The vocabulary pair rule of Gigatoken is closer still, since it is fixed before any input. It certifies cuts before
selected punctuation in SentencePiece input where no vocabulary piece contains the two symbols at the cut in that order,
a sufficient condition argued in a source comment from the fact that a merge's result is always a vocabulary piece
\citep{rod2026gigatoken}. Read under maximal munch, such a pair is a certified window of width two with origin one, by
the reasoning of Lemma~\ref{lem:interior}. The inventory construction here decides certification for every window and
origin through the budget, for an arbitrary regular token set, with a counterexample on each refutation.

\paragraph{The neighboring quantity.} \citet{li2026streaming} statically compute, for a tokenizer, how far past a
token's end a scanner must read before its maximality is settled. Their analysis is finite-state, of the same shape as
Section~\ref{sec:deciding}'s. Among the sources surveyed it is the nearest tokenizer-specific quantity to certification
and the anchor-free span, but it measures confirmation from a start already known. Theorem~\ref{thm:cutoff} instead asks
whether an origin is forced at every occurrence in every context, and Theorem~\ref{thm:density} asks how long an input
can run with no anchor from a given finite inventory. One token set separates them, $T = \{\code{aa}\}$. Its completely
tokenizable inputs are the even-length runs of \code{a}, and no longer token competes with \code{aa}, so nothing has to
be read past a token's end and the confirmation distance they study is zero. Here no occurring window is certified,
because an occurrence of a window of \code{a}s beginning at position $p$ has its last byte covered by the token
beginning at the greatest even position at or below that byte, so the origin is a function of the parity of $p$, and
both parities occur once the input is long enough. No occurring window of any width has a universal origin, the
occurring inventory is empty, and the anchor-free span is unbounded. A bounded confirmation distance therefore does not
imply that any occurring window is certified, and it does not bound the anchor-free span.

\paragraph{Expected density under coverage, worst-case gaps without it.} Two lines of work outside compiler construction
select positions in a byte stream by a bounded window and ask how often the rule selects a position, which is this
paper's supply question asked of a designed rule. \citet{bjorner2010chunking} call a chunking method local when the
cutpoints of every file are exactly the positions whose $h$-vicinity, the $2h + 1$ entries around the position, lies in
a chosen criterion set. A finite certified inventory induces a cut rule of that bounded-context form on every position
whose $h$-vicinity lies inside the file, but the converse fails, since a local rule need not respect token boundaries.
Over $T = \{\code{aa}\}$ the rule selecting every position is local and cuts every token in half, where a certificate
demands the covering-origin condition at every occurrence in every context. Over files whose entries are drawn
independently and uniformly, Bj{\o}rner et~al. define the slack of two files that agree from a point onward as the
distance from that point to their first common cutpoint, analyze its expectation normalized by the expected chunk
length, and separately bound the probability that a long interval carries no cutpoint, for four chunking methods.

\citet{schleimer2003winnowing} define the density of a fingerprinting scheme as the expected fraction of positions
selected under a given input distribution. Assuming independent, uniformly distributed hashes, they prove winnowing's
density asymptotically $2/(w+1)$ for window size $w$, and a lower bound of $1.5/(w+1)$ for every local algorithm, one
that selects a position from each window of $w$ hashes by that window's contents alone. None of these results transfers
here. A certified inventory is under no obligation to select in every window, so the coverage the lower bound assumes
need not hold, and this paper asks for worst-case and corpus gaps, where those results concern expected density under an
independence model.

Two further lines of work study worst-case gaps for arbitrary finite sets of patterns, without tokenization constraints.
A finite set of words is unavoidable when all but finitely many words contain one of them as a factor, which is decided
by building the automaton of the words that avoid the set and looking for a cycle among its live states
\citep{rosaz1995unavoidable}. For a finite unavoidable set over a fixed alphabet, the maximum length of an avoiding word
is bounded in terms of the set's longest word. Rosaz credits this bound to \citet{schutzenberger1964synchronizing},
whose Remark~5 counts the fewest words of length $n$ an unavoidable set over $k$ letters needs.
\citet{crochemore1983unavoidable} showed the bound best possible. A universal hitting set in the sense of
\citet{orenstein2017hitting} is a set of $k$-mers one of which every $L$-long sequence contains, for $L > k$, which is a
bounded anchor-free span for a chosen inventory over all strings.

Theorem~\ref{thm:density} asks this question for a finite inventory of certified windows on completely tokenizable
inputs, with each occurrence marking the origin named by its certificate, and a position is marked once even if several
occurrences witness it. With these constraints the cycle is sought in the credit product of the armed-run verifier
instead of the automaton of the avoiding words. The token set constrains which certificates are admissible, the window
budget and the chosen inventory decide which are used, and Section~\ref{sec:measurements} measures the resulting supply
on named material.

\paragraph{Self-synchronization.} UTF-8's self-synchronization is listed among its characteristics in
\citet{yergeau2003utf8} as ``character boundaries are easily found from anywhere in an octet stream'', resting on the
disjointness of the lead and continuation ranges. The distance is not new either, since the same document bounds a
character at four octets, one lead and at most three continuations. Corollary~\ref{cor:gap} adds that the number is
computed from the token set by a procedure that uses nothing specific to UTF-8 and returns three for that shape. For
token sets with unbounded synchronization distance it returns an unbounded verdict, justified by a reachable cycle of
transitions on non-certified bytes in the live pending automaton.

\paragraph{Subword vocabularies.} The vocabularies measured in Section~\ref{sec:measurements} were produced by byte-pair
encoding in the sense of \citet{sennrich2016subword}, which learns an ordered list of merges and applies them in that
order. \citet{berglund2023formalizing} give that process a formal semantics, a dictionary of ordered rules rewriting a
whole string, and bound how far appending on the right can move the resulting tokenization. Further work gives byte-pair
encoding finite-state and incremental accounts. \citet{berglund2025bpedfa} construct automata operating on BPE
tokenizations together with an input-deterministic transducer relating strings to their tokenizations,
\citet{jiang2026incremental} maintain a BPE tokenization incrementally over every prefix, and
\citet{mamouras2026bpedelay} bound its streaming delay from a known beginning. \citet{cognetta2025tokenization} show
both BPE and MaxMatch representable by finite-state transducers, noting that BPE ``does not tokenize strings from left
to right and requires a notion of priority''. That priority is why Section~\ref{sec:measurements} reads its BPE-derived
vocabularies as maximal-munch token sets and measures that policy instead of BPE's own segmentation. The boundary
comparison emitted here is between depths of the locally trained family under that one reading. A study of agreement
between native merge-rank BPE and MaxMatch is not run, and its outcome is not assumed.

\section{Conclusion}

A certificate is a window and an origin such that, at every occurrence in every completely tokenizable input, the token
covering the window's last byte begins at that origin. Every certified position is a boundary, although not every
boundary is certified. The paper applies this guarantee to splitting, editing, delimiter rules and the design of
grammars and formats. The quantifier that makes recovery certificates trustworthy after damage also supports safe
division of work before it, and deciding certification at inventory scale, with counterexamples, gives these
applications effective constructions.

The limits of the results are stated where each holds. The results apply at the lexical layer, to greedy maximal-munch
scans of literal and regular token sets, and the campaign grammars enter through class abstractions checked by execution
against byte-level scans. The anchor-free span is unbounded for every campaign grammar, and so is the synchronization
distance of each subword vocabulary's certified-byte inventory, so their anchor supply is a corpus fact the instruments
measure, not a constant the theorems guarantee.

The artifact scans with precomputed cuts, no speculation and no reconciliation pass, and asserts the equality of
boundary streams on every run. Beyond this layer, segmentation feeds parsing and caching, and extending these guarantees
there calls for the same sequence of decisions, measurements and executable checks.

\section*{Tools}

Large language models assisted with drafting, with checking citations against primary records and with reviewing the
implementation. The author verified each suggestion by derivation, against primary sources, or against the
implementation, its tests and the archived artifacts, and is responsible for all content.

\appendix
\section{Reproduction and provenance}\label{app:provenance}

\paragraph{Artifact and data availability.} The script \code{sh scripts/check-paper4.sh} regenerates the deterministic
emissions this paper quotes and compares them byte for byte with the committed copies. With the pinned build supplied it
reruns the scanner as well, as described below, and otherwise it checks the prose against the committed scanner
statistics. The same gate rebuilds the manuscript and checks each measurement it names against the emission that
produced it, and it reports that measurements it does not name are not checked. It also checks the runner's environment
file against its schema and the row count required by the declared repetitions. The gate runs in minutes on an ordinary
machine, the wide cutoff sweep being the slowest stage.

The vocabulary measurements read two GPT-2 reference files that are not redistributed here. The gate names them, holds
them to their published digests, and stops with an error when they are absent. The evaluation fixture is not
redistributed either, and the shipped bundle's \code{verify-all.sh} fetches \code{twitter.json} from the pinned simdjson
release and holds it to its digest, while the campaign corpora ship, so the campaign emissions rerun with no network at
all.

Two stages are optional and named as such. Setting \code{CHECK\_PAPER4\_REDERIVE=1} rederives the evaluation,
depth-series and frozen inventories with every decision remade, without reading or writing any committed cache, which
takes hours and writes its own log outside the repository. Passing a built \code{munch} checkout to the gate reruns the
threaded scanner of Section~\ref{sec:artifact}, whose stable lines are compared and whose timing lines are only
recorded. That same stage also reruns the library's own recomputation, described below, and holds every line it
recomputes to the committed record. The emissions, caches and campaign corpora, the programs that derive them, the gate
and the runner are archived together as a standalone record at \zenodo{10.5281/zenodo.23708363}.

The threaded scanner is part of \code{munch}~\citep{munch}, a public C++ library at
\url{https://github.com/nnidhogg/munch} under the MIT license. Its archived record across versions carries one concept
identifier, \zenodo{10.5281/zenodo.21752996}, and the record of release 2.1.0 is \zenodo{10.5281/zenodo.23026012}. The
runner \code{run-artifact.sh} names the repository, the commit it pins, the commit the tag \texttt{v2.1.0} names, and
the sha256 of each probe source, and refuses to run against anything else, which binds the numbers to that commit. The
build directory must sit inside a checkout at that commit, and the probe sources, \code{parallel\_scan.cpp} among them,
must carry their pinned sha256. The commit binds the scanner, since the probe includes the grammar header and links the
library, and a source digest alone would leave most of what runs unaccounted for.

The two GPT-2 reference files are the byte-level GPT-2 tokenizer's vocabulary and merge table, published with that model
as \code{vocab.json} and \code{merges.txt} and mirrored in the model repository \code{openai-community/gpt2}. A reader
places them in a directory named \code{llm-vocabs} beside the bundle root, under the names below, and their digests
decide whether they are the bytes this paper read:

\begin{center}\small
\code{gpt2-vocab.json}\quad\code{196139668be63f3b5d6574427317ae82f612a97c5d1cdaf36ed2256dbf636783}\\*
\code{gpt2-merges.txt}\quad\code{1ce1664773c50f3e0cc8842619a93edc4624525b728b188a9e0be33b7726adc5}
\end{center}

\noindent The gate holds both to those digests before any measurement runs. The measurements read only the merge table,
while the vocabulary file's recorded digest identifies the GPT-2 release that table belongs to.

\paragraph{The decisions in the library.} In release 2.1.0 of \code{munch}, the library's own tests compare three
decisions and one measurement with the exploration programs that model them. The decisions are declared beside the byte
and window certificates and forwarded by its lexer. \code{window\_occurrence()} decides whether a window occurs in some
completely tokenizable input, with the shortest such input. \code{window\_counterexample()} decides whether some
completely tokenizable input violates a window certificate, with the shortest such input, so that an exhausted search
proves the certificate exact. \code{segmentation\_difference()} decides whether two token sets define the same
domain-to-boundary map, with a shortest differing input and whether the difference lies in their domains or their
boundaries. The measurement is the audit tool's supply, read from its report as the anchors the certificates place in a
given input, with their count per KiB and their gaps. At the commit pinned by \code{run-artifact.sh}, the tests compare
each decision's verdict and the length of the shortest input it returns with the exploration that models it. They check
the window decisions on the exploration's own enumerated universes, and the segmentation difference on the 225 literal
pairs of the census in \code{segmentation\_equivalence.py} and the 64 pairs of the eight sets of the full-equivalence
sweep in \code{finitely\_verifiable.py}. The 25 regex pairs and the cross-representation pair are checked by the
exploration programs alone. The supply is checked against the exploration's anchors, supply and gap figures on the
vocabulary and corpus the two share.

The library also carries a probe, \code{paper4\_recompute.cpp}, pinned by \code{run-artifact.sh} beside the scanner,
that rederives this paper's emissions from those decisions in the emissions' own line format. In it the token sets are
rebuilt as the programs build them and digested, the byte certificate is the lexer's split-point test, the window
certificate is \code{window\_counterexample()} exhausted with no counterexample, the supply is the audit report's, the
gap corollary is the anchor-free span, the campaign differential is the boundary difference, and the edit trials replay
the programs' seeded draws through a port of their generator.

The campaign measurements differ from the library's in two respects, their gap convention and the inventory behind the
boundedness verdict. The campaign's gap percentiles include the distances from the input's ends to the nearest anchors,
which the library's supply report omits, and a separate emission line records the library's convention. The campaign
program decides boundedness over the exact class-window inventory, while the probe applies \code{anchor\_free\_span()}
to the byte expansions the conservative window model admits, a subset of that inventory. The probe's unbounded verdict
is consistent with the program's but cannot establish it, since a span unbounded over fewer anchors can become bounded
when anchors are added. On the four campaign rows both are unbounded.

Of the lines the probe recomputes across the eleven emissions, 224 are the committed lines byte for byte. Of the rest,
seven agree up to a remainder the probe does not write, which is either an inventory key that digests the Python sources
or the prose after a verdict. Two agree in their figure under the library's name for the decision, the wide sweep's
refuted and tight counts written with \emph{counterexample} where the emission says \emph{violation}. Four more agree
under the library's spelling of a label. The spellings differ in the unit of the frozen table's header and of the frozen
study's two lines giving the supply as a share of the ceiling, and in the name of the campaign table's block-comment
row. No recomputed line disagrees with its committed line once these remainders, the decision name and the label
spellings are allowed for. The comparison lists the probe's own diagnostics on the campaign and vocabulary runs
separately and excludes them from these counts, since they have no committed counterparts.

The library's exact window decider reproduces the offline verifier's certified-pair counts and the anchors they supply
on the evaluation, calibration, depth-series and campaign universes, and the comparison covers every width of the budget
on the evaluation slice and every certified pair individually on the conventional campaign row. The eighteen lines the
probe does not recompute are prose, the attribution comments, the line naming the trial records with their digest, and
seven lines of the wide sweep. Those seven report five counts from the exploration's checks of the cutoff proof's two
truncations. These checks enumerate to a declared margin above each cutoff and are not implemented in the library. The
exploration programs generated the committed emissions and remain the independent reference the library is tested
against, while the library's recomputation checks those emissions in return.

The library audit report's window row counts the windows its conservative window model certifies. With accepted proper
token prefixes that model can reject windows the offline verifier certifies, so on a byte-fallback vocabulary the row
gives a lower bound on the exact inventory of Table~\ref{tab:supply}, while the certified-byte counts agree exactly.

\setlength{\bibsep}{1pt}
\bibliographystyle{plainnat}
\bibliography{refs}

\end{document}